\documentclass[letterpaper,11pt]{article}
\usepackage[utf8]{inputenc}

\pdfoutput=1

\usepackage{thmtools}
\usepackage{thm-restate}
\usepackage{amsmath, amsthm, amssymb}
\usepackage{algpseudocode,algorithm}
\usepackage{algorithmicx}
\usepackage{mathtools}
\usepackage[numbers]{natbib}
\usepackage{comment} 
\usepackage{color-edits} 
\usepackage[most]{tcolorbox}
\usepackage{xfrac}
\usepackage{hyperref}
\usepackage{multirow}
\usepackage{caption}
\usepackage{bm}
\usepackage{newfloat}
\usepackage{enumitem}
\usepackage[utf8]{inputenc}
\usepackage{footmisc}
\usepackage{xpatch}

\usepackage{thmtools,thm-restate} 

\makeatletter
\xpatchcmd{\algorithmic}
  {\ALG@tlm\z@}{\leftmargin\z@\ALG@tlm\z@}
  {}{}
\makeatother

\usepackage[T1]{fontenc}

\usepackage[margin=1in]{geometry}

\allowdisplaybreaks

\definecolor{mygreen}{RGB}{20,120,60}

\title{Near-Optimal Massively Parallel Graph Connectivity\footnote{A preliminary version of this paper is to appear in the proceedings of {\em The 60th Annual IEEE Symposium on Foundations of Computer Science }(FOCS 2019).\vspace{0.2cm}}}

\author{Soheil Behnezhad\thanks{Supported in part by NSF SPX grant CCF-1822738 and a Google PhD Fellowship.} \\ University of Maryland \\ \texttt{soheil@cs.umd.edu}
\and Laxman Dhulipala \\ Carnegie Mellon University \\ \texttt{ldhulipa@cs.cmu.edu}
\and Hossein Esfandiari \\ Google Research \\ \texttt{esfandiari@google.com}
\and Jakub Łącki \\ Google Research \\ \texttt{jlacki@google.com}
\and Vahab Mirrokni \\ Google Research \\ \texttt{mirrokni@google.com}
}

\date{}

\newcommand{\MPC}[0]{\ensuremath{\mathsf{MPC}}}

\newcommand{\PRAM}[0]{\ensuremath{\mathsf{PRAM}}}

\newcommand{\scanpram}[0]{\emph{scan} \ensuremath{\mathsf{PRAM}}\renewcommand{\scanpram}[0]{scan \ensuremath{\mathsf{PRAM}}}}
\newcommand{\arbpram}[0]{\emph{arbitrary} CRCW \ensuremath{\mathsf{PRAM}}\renewcommand{\arbpram}[0]{arbitrary \ensuremath{\mathsf{PRAM}}}}
\newcommand{\pripram}[0]{\emph{priority} CRCW \ensuremath{\mathsf{PRAM}}\renewcommand{\pripram}[0]{priority CRCW \ensuremath{\mathsf{PRAM}}}}
\newcommand{\combpram}[0]{\emph{combining} CRCW \ensuremath{\mathsf{PRAM}}\renewcommand{\combpram}[0]{combining CRCW \ensuremath{\mathsf{PRAM}}}}

\newcommand{\multipram}[0]{\emph{multiprefix} CRCW \ensuremath{\mathsf{PRAM}}\renewcommand{\multipram}[0]{multiprefix CRCW \ensuremath{\mathsf{PRAM}}}}
\newcommand{\twocycle}[0]{\textsc{2-Cycle}}

\DeclareMathOperator{\dist}{dist}

\let\next\relax
\DeclareMathOperator{\next}{next}

\newcommand{\level}[1]{\ensuremath{\ell(#1)}}
\newcommand{\budget}[1]{\ensuremath{b(#1)}}

\newcommand{\budgetconstant}{\ensuremath{1.25}}

\newcommand{\eqdef}[0]{\ensuremath{:=}}

\newcommand{\newpara}[1]{\smallskip\noindent {\bf #1.}}

\DeclareMathOperator{\poly}{poly}

\newcommand{\E}[0]{\ensuremath{\mathbb{E}}}

\newcommand{\NC}[0]{\ensuremath{\mathsf{NC}}}
\renewcommand{\P}[0]{\ensuremath{\mathsf{P}}}

\addauthor{sb}{blue}    
\addauthor{md}{red}  
\addauthor{ss}{purple} 

\newtheorem{theorem}{Theorem}
\newtheorem{lemma}{Lemma}[section]

\newtheorem{claim}[lemma]{Claim}

\newtheorem{observation}[lemma]{Observation}
\newtheorem{remark}[lemma]{Remark}

\newtheorem{conjecture}[lemma]{Conjecture}

\definecolor{mygreen}{RGB}{20,140,80}
\definecolor{linkcolor}{RGB}{0,0,230}
\definecolor{mylightgray}{RGB}{230,230,230}
\definecolor{verylightgray}{RGB}{245,245,245}

\hypersetup{
     colorlinks=true,
     citecolor= mygreen,
     linkcolor= mygreen,
     urlcolor= mygreen
}

\newcommand{\etal}[0]{\textit{et al.}}

\algnewcommand{\IIf}[1]{\State\algorithmicif\ #1\ \algorithmicthen}
\algnewcommand{\EndIIf}{\unskip\ \algorithmicend\ \algorithmicif}

\newcommand{\smparagraph}[1]{\vspace{-0.3cm}\paragraph{#1}}

\newcounter{myalgctr}

\newtcolorbox{OuterBox}[1][]{%
    breakable,
    enhanced,
    frame hidden,
    interior hidden,
    left=-5pt,
    right=-5pt,
    top=-5pt,
    float=p,
    boxsep=0pt,
    arc=0pt
#1}%

\newtcolorbox{InnerBox}[1][]{%
    enforce breakable,
    enhanced,
    colback=gray,
    colframe=white,
#1}%

\newenvironment{tbox}{
\begin{tcolorbox}[width=\textwidth,
                  enhanced,
                  boxsep=2pt,
                  left=1pt,
                  right=1pt,
                  top=4pt,
                  boxrule=1pt,
                  arc=0pt,
                  colback=white,
                  colframe=black,
	              breakable,
	              float=t
                  ]
}{
\end{tcolorbox}
}

\newenvironment{graytbox}{
\vspace{0.2cm}
\begin{tcolorbox}[width=\textwidth,
                  enhanced,
                  frame hidden,
                  boxsep=6pt,
                  left=1pt,
                  right=1pt,
                  top=4pt,
                  boxrule=1pt,
                  arc=0pt,
                  colback=mylightgray,
                  colframe=black,
                  breakable
                  ]
}{
\end{tcolorbox}
}

\newcommand{\tboxhrule}[0]{\vspace{0.1cm} {\color{black} \hrule} \vspace{0.2cm}}

\newenvironment{titledtbox}[1]{\begin{tbox}#1 \tboxhrule}{\end{tbox}}

\newenvironment{tboxalg}[1]{\refstepcounter{myalgctr}\begin{titledtbox}{\textbf{Algorithm \themyalgctr.} #1}}{\end{titledtbox}}

\begin{document}
\begin{titlepage}
\maketitle
\thispagestyle{empty}

\begin{abstract}
\setlength{\parskip}{0.4em}
Identifying the connected components of a graph, apart from being a fundamental problem with countless applications, is a key primitive for many other algorithms. In this paper, we consider this problem in parallel settings. Particularly, we focus on the {\em Massively Parallel Computations} (\MPC{}) model, which is the standard theoretical model for modern parallel frameworks such as MapReduce, Hadoop, or Spark. We consider the {\em truly sublinear} regime of \MPC{} for graph problems where the space per machine is $n^\delta$ for some desirably small constant $\delta \in (0, 1)$.

We present an algorithm that for graphs with diameter $D$ in the wide range $[\log^{\epsilon} n, n]$, takes $O(\log D)$ rounds to identify the connected components and takes $O(\log \log n)$ rounds for all other graphs. The algorithm is randomized, succeeds with high probability\footnote{We use the term {\em with high probability} to refer to probability at least $1-n^{-c}$ for arbitrarily large constant $c>1$.}, does not require prior knowledge of $D$, and uses an optimal total space of $O(m)$. We complement this by showing a conditional lower-bound based on the widely believed \twocycle{} conjecture that $\Omega(\log D)$ rounds are indeed necessary in this setting.

Studying parallel connectivity algorithms received a resurgence of interest after the pioneering work of Andoni~\etal{}~[FOCS 2018] who presented an algorithm with $O(\log D \cdot \log \log n)$ round-complexity. Our algorithm improves this result for the whole range of values of $D$ and almost settles the problem due to the conditional lower-bound.

Additionally, we show that with minimal adjustments, our algorithm can also be implemented in a variant of the (CRCW) \PRAM{} in asymptotically the same number of rounds.\end{abstract}

\end{titlepage}

\clearpage

\section{Introduction}
Identifying the connected components of a graph is a fundamental problem that has been studied in a variety of settings (see e.g. \cite{DBLP:conf/soda/AhnGM12, DBLP:journals/tcs/FeigenbaumKMSZ05, DBLP:conf/esa/CrouchMS13, DBLP:journals/jal/ShiloachV82, DBLP:journals/dam/Vishkin84, icdepaper} and the references therein). This problem is also of great practical importance~\cite{vldbsurvey} with a wide range of applications, e.g. in clustering~\cite{icdepaper}. The main challenge is to compute connected components in graphs with over hundreds of billions or even trillions of nodes and edges~\cite{vldbsurvey, wsdm}. The related theory question is: 
\begin{center}
{\em What is the true complexity of finding connected components in massive graphs?}
\end{center}
We consider this problem in parallel settings which are a common way of handling massive graphs. Our main focus is specifically on the {\em Massively Parallel Computations} (\MPC{}) model  \cite{DBLP:journals/jacm/BeameKS17, DBLP:conf/soda/KarloffSV10, DBLP:conf/isaac/GoodrichSZ11}; however, we show that our techniques are general enough to be seamlessly implemented in other parallel models such as (CRCW) \PRAM{}. The \MPC{} model is arguably the most popular theoretical model for modern parallel frameworks such as MapReduce~\cite{DBLP:journals/cacm/DeanG08}, Hadoop~\cite{Hadoop}, or Spark~\cite{DBLP:conf/hotcloud/ZahariaCFSS10} and has received significant attention over the past few years (see Section~\ref{sec:relatedwork}). We consider the strictest regime of \MPC{} for graph problems where the space per machine is strongly sublinear in $n$.

\smparagraph{The \MPC{} model.} The input, which in our case is the edge-set of a graph $G(V, E)$ with $n$ vertices and $m$ edges, is initially distributed across $M$ machines. Each machine has a space of size $S = n^\delta$ words where constant $\delta \in (0, 1)$ can be made arbitrarily small. Furthermore, $M \cdot S = O(m)$ so that there is only enough total space to store the input. Computation proceeds in rounds. Within a round, each machine can perform arbitrary computation on its local data and send messages to each of the other machines. Messages are delivered at the beginning of the following round. An important restriction is that the total size of messages sent and received by each machine in a round should be $O(S)$. The main objective is to minimize the number of rounds that are executed.

\smparagraph{What we know.} Multiple algorithms for computing connected components in $O(\log n)$ \MPC{} rounds have been shown~\cite{DBLP:conf/soda/KarloffSV10, wsdm, icdepaper, cc-contractions}. On the negative side, a popular \twocycle{} conjecture~\cite{DBLP:conf/icml/YaroslavtsevV18, DBLP:conf/spaa/RoughgardenVW16, cc-contractions, DBLP:journals/corr/abs-1805-02974} implies that $\Omega(\log n)$ rounds are necessary. Namely, the conjecture states that in this regime of \MPC{}, distinguishing between a cycle on $n$ vertices and two cycles on $n/2$ vertices each requires $\Omega(\log n)$ rounds. However, the \twocycle{} conjecture and the matching upper bound are far from explaining the true complexity of the problem. First, the hard example used in the conjecture is very different from what most graphs look like. Second, the empirical performance of the existing algorithms (in terms of the number of rounds) is much lower than what the upper bound of $O(\log n)$ suggests~\cite{cc-beyond, cc-contractions, wsdm, icdepaper, DBLP:journals/tpds/LulliCDLR17}.

This disconnect between theory and practice has motivated the study of graph connectivity as a function of diameter $D$ of the graph. The reason is that the vast majority of real-world graphs, indeed have very low diameter~\cite{snapnets, DBLP:journals/tcs/CrescenziGHLM13}. This is reflected in multiple theoretical models designed to capture real-world graphs, which yield graphs with polylogarithmic diameter~\cite{DBLP:journals/combinatorica/BollobasR04, gu2013clustering, lu2001diameter, bollobas2003mathematical}.

\smparagraph{Our contribution.} Our main contribution is the following algorithm:

\begin{graytbox}
	\begin{theorem}[main result]\label{thm:main}
		There is a strongly sublinear \MPC{} algorithm with $O(m)$ total space that given a graph with diameter $D$, identifies its connected components in $O(\log D)$ rounds if $D \geq \log^{\epsilon} n$ for any constant $\epsilon >0$, and takes $O(\log\log_{m/n} n)$ rounds otherwise. The algorithm is randomized, succeeds with high probability, and does not require prior knowledge of $D$.
	\end{theorem}
\end{graytbox}

The \twocycle{} conjecture mentioned above directly implies that $o(\log D)$ round algorithms do not exist in this setting for $D=\Theta(n)$. However, it does not rule out the possibility of achieving an $O(1)$ round algorithm if e.g. $D = O(\sqrt{n})$. We refute this possibility and show that indeed for any choice of $D \in [\log^{1+\Omega(1)}, n]$, there are family of graphs with diameter $D$ on which $\Omega(\log D)$ rounds are necessary in this regime of \MPC{}, if the \twocycle{} conjecture holds.

\begin{restatable}{theorem}{hard}\label{thm:hard}
	Fix some $D' \geq \log^{1+\rho} n$ for a desirably small constant $\rho \in (0, 1)$. Any \MPC{} algorithm with $n^{1-\Omega(1)}$ space per machine that w.h.p. identifies each connected component of any given $n$-vertex graph with diameter $D'$ requires $\Omega(\log D')$ rounds, unless the \twocycle{} conjecture is wrong.
\end{restatable}

We note that proving any unconditional super constant lower bound for any problem in $\P$, in this regime of \MPC{}, would imply $\NC^1 \subsetneq \P$ which seems out of the reach of current techniques \cite{DBLP:conf/spaa/RoughgardenVW16}.

\smparagraph{Extention to \PRAM{}.} As a side result, we provide an implementation of our connectivity algorithm in $O(\log D + \log \log_{m/n} n)$ depth in the multiprefix CRCW \PRAM{} model, a parallel computation model that permits concurrent reads and concurrent writes. This implementation of our algorithm performs $O((m + n)(\log D + \log \log_{m/n} n))$ work and is therefore \emph{nearly} work-efficient. The following theorem states our result. We defer further elaborations on this result to Appendix~\ref{apx:paralleconnimpl}.

\begin{restatable}{theorem}{workdepthpram} \label{cl:workdepthpram}
There is a \multipram{} algorithm that given a graph with diameter $D$, identifies its connected components in $O(\log D + \log \log_{m/n} n)$ depth and $O((m+n)(\log D + \log\log_{m/n} n))$ work.
The algorithm is randomized, succeeds with high probability and does not require prior knowledge of $D$.
\end{restatable}

\smparagraph{Comparison with the state-of-the-art.} The round complexity of our algorithm improves over that of the state-of-the-art algorithm by Andoni {\em et al.}~\cite{andoniparallel} that takes $O(\log D \cdot \log\log_{m/n} n)$ rounds. Note that the algorithm of \cite{andoniparallel} matches the $\Omega(\log D)$ lower bound for a very specific case: if the graph is extremely dense, i.e., $m = n^{1+\Omega(1)}$. In practice, this is usually not the case~\cite{chung2010graph,snapnets, farkas2001spectra}. In fact, it is worth noting that the main motivation behind the \MPC{} model with sublinear in $n$ space per machine is the case of sparse graphs \cite{DBLP:conf/soda/KarloffSV10}. We also note that for the particularly important case when $D = \poly \log n$, our algorithm requires only $O(\log \log n)$ rounds. This improves quadratically over a bound of $O(\log^2 \log n)$ rounds, which follows from the result of \cite{andoniparallel}.

Our result also provides a number of other qualitative advantages. For instance it succeeds with high probability
as opposed to the constant success probability of \cite{andoniparallel}. Furthermore, the running time required for identifying each connected component depends on its own diameter only. The diameter $D$ in the result of \cite{andoniparallel} is crucially the largest diameter in the graph. 

\subsection{Further Related work}\label{sec:relatedwork}

The \MPC{} model has been extensively studied over the past few years especially for graph problems. See for instance \cite{DBLP:conf/soda/KarloffSV10, DBLP:conf/isaac/GoodrichSZ11, DBLP:conf/spaa/LattanziMSV11, DBLP:conf/pods/BeameKS13, DBLP:conf/stoc/AndoniNOY14, DBLP:conf/spaa/AhnG15, DBLP:conf/spaa/RoughgardenVW16, DBLP:conf/stoc/ImMS17, DBLP:conf/nips/BateniBDHKLM17, DBLP:conf/stoc/CzumajLMMOS18, DBLP:conf/podc/GhaffariGKMR18, DBLP:conf/spaa/BehnezhadDETY17, DBLP:journals/corr/abs-1802-10297, DBLP:conf/icalp/BateniBDHM18, assadiedcs, DBLP:journals/corr/abs-1807-06701, DBLP:journals/corr/abs-1807-06251, DBLP:journals/corr/abs-1807-05374, andoniparallel, DBLP:journals/corr/abs-1805-02974, DBLP:conf/soda/BoroujeniEGHS18} and the references therein.

More relevant to our work on graph connectivity, a recent result by Assadi \emph{et al.}~\cite{DBLP:journals/corr/abs-1805-02974} implies an $O(\log \frac{1}{\lambda} + \log \log n)$ round algorithm for graphs with $\widetilde{O}(n)$ edges that have {\em spectral gap} at least $\lambda$. By a well-known bound, $D = O(\frac{\log n}{\lambda})$. Therefore, our algorithm requires $O(\log D + \log \log n) = O(\log (\frac{\log n}{\lambda}) + \log \log n) = O(\log \frac{1}{\lambda} + \log \log n)$ rounds for graphs with spectral gap at least $\lambda$. As a result, the running time bound of our algorithm is never worse than the bound due to Assadi \emph{et al.} However, as shown in \cite{andoniparallel}, there are graphs with $\frac{1}{\lambda} \geq D \cdot n^{\Omega(1)}$ making our algorithm more general.

Finally, we remark that a preprint claiming a deterministic \MPC{} connectivity algorithm requiring only $O(\log D)$ rounds has been published recently~\cite{DBLP:journals/corr/abs-1808-06705}. However, the key claim of the paper is fundamentally incorrect. Specifically, the paper first shows that the algorithm requires $O(\log k)$ rounds for a graph that is a path of length $k$ and directly concludes that this implies round complexity of $O(\log D)$ (see Lemma 3 in~\cite{DBLP:journals/corr/abs-1808-06705}). This kind of reasoning is not valid. In particular, the HashToMin algorithm~\cite{icdepaper} works in $O(\log k)$ rounds for graphs consisting of disjoint paths of length at most $k$, but has been shown to require $\omega(\log D)$ rounds for certain graphs~\cite{andoniparallel}. Apart from the fact that the proof is incorrect, we believe that the counterexample of~\cite{andoniparallel} also shows that the algorithm of~\cite{DBLP:journals/corr/abs-1808-06705} on a family of graphs with diameter $O(\log n)$, has round complexity $\Omega(\log n) = \Omega(D)$.

\subsection{Paper Organization}
Section~\ref{sec:highlevel} gives a high-level overview of our techniques.
Sections~\ref{sec:algorithm} and \ref{sec:shrinkvertices} are devoted to our algorithm, its correctness and performance.
Then, in Section~\ref{sec:hardness} we give a new lower bound for the problem of solving connectivity in the \MPC{} model.
In Appendix~\ref{apx:implementation} we describe the implementation details of the algorithm in the \MPC{} model. We remark that the implementation follows from standard techniques, but we provide it for completeness.

\section{High-Level Overview of Techniques}\label{sec:highlevel}
Recall that we assume the regime of \MPC{} with strictly sublinear space of $n^\delta$ with $\delta$ being a constant in $(0, 1)$. This local space, roughly speaking, is usually not sufficient for computing any meaningful global property of the graph within a machine. As such, most algorithms in this regime proceed by performing \emph{local operations} such as contracting edges/vertices, adding edges, etc. Note that even the direct neighbors of a high-degree vertex may not fit in the memory of a single machine, however, using standard techniques most of basic local operations can be implemented in $O(1/\delta)=O(1)$ rounds of \MPC{}. The details are given in Appendix~\ref{apx:implementation}. For the purpose of this section, we do not get into technicalities of how this can be done.

We start with a brief overview of some of the relevant techniques and results, then proceed to describe the new ingredients of our algorithm and its analysis.

\smparagraph{Graph exponentiation.} Consider a simple and well-known algorithm that connects every vertex to vertices within its 2-hop (i.e., vertices of distance 2) by adding edges. It is not hard to see that the distance between any two vertices shrinks by a factor of $2$. By repeating this procedure, each connected component becomes a clique within $O(\log D)$ steps. The problem with this approach, however, is that the total space required can be up to $\Omega(n^2)$, which for sparse graphs is much larger than $O(m)$.

Andoni \etal{}~\cite{andoniparallel} manage to improve the total space to the optimal bound of $O(m)$, at the cost of increasing the round complexity to $O(\log D \cdot \log \log_{m/n} n)$. We briefly overview this result below.

\smparagraph{Overview of Andoni \etal{}'s algorithm.} Suppose that {\em every} vertex in the graph has degree at least $d \gg \log n$. Select each vertex as a {\em leader} independently with probability $\Theta(\log n/d)$. Then contract every non-leader vertex to a leader in its 1-hop (which w.h.p. exists). This shrinks the number of vertices from $n$ to $\widetilde{O}(n/d)$. As a result, the amount of space available per remaining vertex increases to $\widetilde{\Omega}(\frac{m}{n/d}) = \widetilde{\Omega}(\frac{n d}{n/d}) \approx d^2$. At this point, a variant of the aforementioned graph exponentiation technique can be used to increase vertex degrees to $d^2$ (but not more), which implies that another application of leader contraction decreases the number of vertices by a factor of $\Omega(d^2)$. Since the available space per remaining vertex increases doubly exponentially, $O(\log \log n)$ phases of leader contraction suffice to increase it to $n$ per remaining vertex. Moreover, each phase requires $O(\log D)$ iterations of graph exponentiation, thus the overall round complexity is $O(\log D \log \log n)$.

\subsection{Our Connectivity Algorithm: The Roadmap}
The main shortcoming of Andoni \etal{}'s algorithm is that within a phase, where the goal is to increase the degree of every vertex to $d$, those vertices that have already reached degree $d$ are {\em stalled} (i.e., do not connect to their 2-hops) until all other vertices reach this degree. Because of the stalled vertices, the only guaranteed outcome of the graph exponentiation operation is increasing vertex degrees.
In particular, the diameter of the graph may remain unchanged. This is precisely why their algorithm may require up to $O(\log D \cdot \log \log n)$ applications of graph exponentiation. We note that this is not a shortcoming of their analysis. Indeed, we remark that there are family of graphs on which Andoni \etal{}'s algorithm takes $\Theta(\log D \cdot \log \log n)$ rounds.

Instead of describing our algorithm, we focus in this section on some of the properties that we expect it to satisfy, and how they suffice to get our desired round complexity. This overview should be helpful when reading the description of the algorithm in Section~\ref{sec:algdescribe}.

Our algorithm assigns \emph{budgets} to vertices.
Intuitively, a budget controls how much space a vertex can use, i.e., how much it can increase its degree.
To bound the space complexity, we will bound the sum of all vertex budgets.
In our algorithm vertices may have different budgets (differently from the algorithm of Andoni \etal{}).
This allows us to prevent the vertices from getting stalled behind each other. Overall, we have $L = \Theta(\log \log n)$ possible budgets $\beta_0, \beta_1, \ldots, \beta_L$ where $\beta_0 = O(1)$, $\beta_L = \Omega(n)$, and $\beta_i = (\beta_{i-1})^{\budgetconstant}$. We say a vertex $v$ is at {\em level} $\ell(v) = i$, if its budget $b(v)$ equals $\beta_i$.
The algorithm executes a single loop until each connected component
becomes a clique. We call a single execution of this loop an~\emph{iteration} which can be implemented in $O(1)$ rounds of \MPC{}.

\vspace{0.2cm}
\noindent \textbf{Property 1} (see Lemma~\ref{lem:2hoporlevelincrease} for a formal statement)\textbf{.} For any two vertices $u$ and $v$ at distance exactly 2 at the beginning of an iteration of the algorithm, after the next 4 iterations, either their distance decreases to 1, or the level of both vertices increases by at least one.
\vspace{0.2cm}

We call every 4 iterations of the algorithm a {\em super-iteration}. Property 1 guarantees that if a vertex does not get connected to every vertex in its 2-hop within a super-iteration, its level must increase.\footnote{In this section, for simplicity and to convey the main intuitions, we ignore the changes to the graph's structure caused by vertices being merged together.} Recall, on the other hand, that the maximum level of any vertex is at most $O(\log \log n)$. As such, every vertex resists getting connected to those in its 2-hop for at most $O(\log \log n)$ super-iterations. However, somewhat counter-intuitively, this observation is (provably) not sufficient to guarantee an upper bound of $O(\log D + \log \log n)$ rounds. Our main tool in resolving this, is maintaining another property.

\vspace{0.2cm}
\noindent \textbf{Property 2} (see Observation~\ref{obs:levellargerthanneighbors} for a formal statement)\textbf{.} If a vertex $v$ is neighbor of a vertex $u$ with $\ell(u) > \ell(v)$, then by the end of the next iteration, the level of $v$ becomes at least $\ell(u)$.
\vspace{0.2cm}

The precise proof of sufficiency of Properties 1 and 2 is out of the scope of this section. Nonetheless, we provide a proof sketch with the hope to make the actual arguments easier to understand. See Lemma~\ref{lem:lplusloglog} for the formal statement and its proof.

\begin{proof}[Proof sketch of round complexity.] Fix two vertices $u$ and $v$ in one connected component of the original graph and let $P_1$ be the shortest path between them. As the vertices connect to their 2-hops and get closer to each other, we drop some of the vertices of $P_1$ while ensuring that the result is also a valid path from $u$ to $v$. More precisely, we maintain a path $P_i$ by the end of each super-iteration $i$ which is obtained by replacing some subpaths of $P_{i-1}$ of length at least two by single edges. We say that the interior vertices of the replaced subpaths are dropped.

Our goal is to show that for $R := O(\log D + \log \log n)$, path $P_R$ has $O(1)$ length, thus dropping the diameter of the whole graph to $O(1)$ which is trivially solved in $O(1)$ rounds by our algorithm. To show this, we require a potential-based argument. Suppose that we initially put one coin on every vertex of $P_1$, thus we have at most $D+1$ coins. Let $P_i = (u_1, \ldots, u_j, \ldots, u_k)$ be the path at the end of super-round $i$. As we construct $P_{i+1}$ from $P_i$, any vertex $u_j$ that is dropped from $P_i$, passes its coins evenly to vertices in $\{u_{j-2}, u_{j-1}, u_{j+1}, u_{j+2}\}$ that exist and survive to $P_{i+1}$ (if none of them survive the coins are discarded).\footnote{In an earlier version of this paper, we passed on the coins only to direct neighbors in the path. As pointed out by \cite{DBLP:journals/corr/abs-2003-00614} this may cause a parity-type problem for one of the last two vertices in the path. The authors in \cite{DBLP:journals/corr/abs-2003-00614} already propose another fix. But passing on the coins to the two vertices to the right and two vertices to the left, as we do in this version of the paper, also easily resolves this issue. We note that this change appears only in the analysis, and the algorithm and all the claimed bounds remain unchanged.\label{fnlabel}} Moreover, we construct $P_{i+1}$ from $P_i$ such that it satisfies the following property, which we call invariant 1: If the level of a vertex $u_j \in P_i$ within super-iteration $i+1$ does not increase, there is a vertex in $\{u_{j-2}, u_{j-1}, u_j, u_{j+1}, u_{j+2}\}$ that is dropped. This is guaranteed to be possible due to Property~1: Observe that at least one of $u_{j-2}$ and $u_{j+2}$ should belong to $P_i$ otherwise the path has length $\leq 2$ and is already small. Let us suppose w.l.o.g. that $u_{j-2}$ exists. Now if either of $u_{j-2}$ or $u_j$ is dropped from $P_i$ the condition is satisfied, otherwise by Property~1, $u_{j-2}$ and $u_j$ should be directly connected after super-iteration $i+1$ and  and it is thus safe to drop $u_{j-1}$ and ensure $P_{i+1}$ remains to be a path.

Finally, we use Property 2 to prove invariant 2: In any path $P_i$, every vertex of level $j$ has at least $(1.25)^{i - j}$ coins. That is, we have {\em more} coins on the vertices that have {\em lower} levels. Note that this is sufficient to prove the round complexity. For, otherwise, if $|P_R| > 2$, due to the fact that the level of every vertex is at most $O(\log \log n)$, there should remain a vertex in $P_R$ with at least
$$
(1.25)^{R-j} \geq (1.25)^{\Theta(\log D + \log\log n) - O(\log\log n)} \geq (1.25)^{4\log D} \gg D+1
$$
coins, while we had only $D+1$ coins in total.  Property 2 is useful in the proof of invariant 2 in the following sense: If a low-level vertex $w \in P_i$ survives to $P_{i+1}$ without increasing its level, its dropped 2-hop neighbor (which exists by invariant 1) cannot have a higher level than $w$ by Property~2 (since their distance is at most two and a super-iteration includes four iterations), and thus passes enough coins to $w$.
\end{proof}

\section{Main Algorithm: Connectivity with $O(m + n\poly\log n)$ Total Space}\label{sec:algorithm}

In this section, we describe an $O(\log D + \log\log_{T/n} n)$ round connectivity algorithm assuming that the total available space is $T \geq m + n \log^\alpha n$ where $\alpha$ is some desirably large constant. We later show how to improve the total space to the optimal bound of $O(m)$ in Section~\ref{sec:shrinkvertices}. We start with description of the algorithm in Section~\ref{sec:algdescribe} and proceed to its analysis in Sections~\ref{sec:alganalysis}, \ref{sec:round}, and \ref{sec:space}.

\begin{remark}
	For simplicity, we describe an algorithm that succeeds with probability $1-1/n^{\Omega(1)}$. One can boost the success probability to $1-1/n^c$ by changing some parameters in the algorithm or by simply running $O(c)$ independent copies of it in parallel.
\end{remark}

\subsection{The Algorithm}\label{sec:algdescribe}

\newcommand{\connectivity}[0]{\mathsf{FindConnectedComponents}}
\newcommand{\addtwohop}[0]{\ensuremath{\mathsf{Connect2Hops}}}
\newcommand{\relabelinter}[0]{\ensuremath{\mathsf{RelabelInterLevel}}}
\newcommand{\relabelintra}[0]{\ensuremath{\mathsf{RelabelIntraLevel}}}
\newcommand{\applyrelabelings}[0]{\ensuremath{\mathsf{ApplyRelabelings}}}

The algorithm consists of a number of \emph{iterations}, each of which calls three subroutines named $\addtwohop$, $\relabelinter$, and $\relabelintra$.\footnote{We note that the relabeling subroutines are close to the leader contraction operation we discussed in Section~\ref{sec:highlevel}. However, we use a different terminology to emphasize the difference in handling {\em chains}. See Figure~\ref{fig:chains}.} Each iteration will be implemented in $O(1)$ rounds of \MPC{} and we later show that $O(\log D + \log\log_{T/n} n)$ iterations are sufficient. We first formalize the overall structure of the algorithm as Algorithm~\ref{alg:connectivity}, then continue to describe the subroutines of each iteration one by one. 

\begin{tboxalg}{$\connectivity(G(V, E))$}\label{alg:connectivity}
	\begin{enumerate}[ref={\arabic*}, topsep=0pt,itemsep=0ex,partopsep=0ex,parsep=1ex, leftmargin=*]
		\item To any vertex $v$, we assign a {\em level} $\level{v} \gets 0$, a {\em budget} $\budget{v} \gets (\frac{T}{n})^{1/2}$, and a set $C(v) \gets \{ v \}$ which throughout the algorithm, denotes the set of vertices that $v$ {\em corresponds to}. Moreover, every vertex is initially marked as {\em active} and we set $\next(v) \gets v$ for every vertex $v$.\label{line:initialize}
		\item Repeat the following steps until each remaining connected components becomes a clique:
			\begin{enumerate}[ref={(\alph*)}, topsep=5pt,itemsep=0ex]
				\item $\addtwohop{}(G, b, \ell)$
				\item $\relabelinter{}(G, b, \ell, C, \next)$
				\item $\relabelintra{}(G, b, \ell, C)$
			\end{enumerate}
		\item For every remaining connected component $\mathcal{C}$ corresponds to one of the connected components of the original graph whose vertex set is $\cup_{v \in \mathcal{C}} C(v)$.
	\end{enumerate}
\end{tboxalg}
\begin{titledtbox}{$\addtwohop{}(G, b, \ell)$}
	For every \underline{active} vertex $v$:
	\begin{enumerate}[topsep=2pt,itemsep=0ex,partopsep=0ex,parsep=1ex]
		\item Define $H(v) \eqdef \{u \, | \, \exists w \text{ s.t. } w \in N(u) \cap N(v), \level{u} = \level{w} = \level{v}\}$.
		\item Let $d_v$ be the number of vertices currently connected to $v$ that have level at least $\level{v}$. Pick $\min\{ \budget{v} - d_v, |H(v)|\}$ arbitrary vertices in $H(v)$ and connect them to $v$.
	\end{enumerate}
\end{titledtbox}

\begin{titledtbox}{$\relabelinter(G, b, \ell, C, \next)$}
	\begin{enumerate}[ref={\arabic*}, topsep=0pt,itemsep=-0.5ex,partopsep=0ex,parsep=1ex, leftmargin=*]
		\item For every (active or inactive) vertex $v$:
		\begin{enumerate}[topsep=-1pt,itemsep=-1ex,partopsep=0ex,parsep=1ex]
			\item Let $h(v)$ be the neighbor of $v$ with the highest level with ties broken arbitrarily.
			\item If $\level{h(v)} > \level{v}$, mark $v$ as {\em inactive} and set $\next(v) \gets h(v)$.
		\end{enumerate}
		\item Replace every edge $\{u, v\}$ in the graph with $\{\next(u), \next(v)\}$.
		\item Remove duplicate edges and self-loops. \label{line:removeduplicateandselfloops}
		\item For every vertex $v$, set $I(v) \gets \cup_{u: \next(u)=v} C(u)$.
		\item For every vertex $v$, if $v$ is active, set $C(v) \gets C(v) \cup I(v)$ and if $v$ is inactive set $C(v) \gets I(v)$.\label{line:passc}
		\item If an inactive vertex has become isolated, remove it from the graph.
	\end{enumerate}
\end{titledtbox}
\begin{titledtbox}{$\relabelintra(G, b, \ell, C)$}
	\begin{enumerate}[ref={\arabic*}, topsep=0pt,itemsep=-0.5ex,partopsep=0ex,parsep=1ex, leftmargin=*]
		\item Mark an active vertex $v$ as ``saturated''  if it has more than $\budget{v}$ active  neighbors that have the same level as $v$.\label{line:saturated1} 
		\item If an active vertex $v$ has a neighbor $u$ with $\ell(u) = \ell(v)$ that is marked as saturated, $v$ is also marked as saturated.\label{line:saturated2}
		\item Mark every saturated vertex $v$ as a ``leader'' independently with probability $\min\{\frac{3\log n}{\budget{v}}, 1\}$.
		\item \label{line:levelincrease} For every leader vertex $v$, set $\level{v} \gets \level{v}+1$ and $\budget{v} \gets \budget{v}^{\budgetconstant}$.
		\item Every non-leader saturated vertex $v$ that sees a leader vertex $u$ of the same level (i.e., $\level{u}=\level{v}$) in its 2-hop (i.e., $\dist(v, u) \leq 2$), chooses one as its leader arbitrarily.
		\item Every vertex is contracted to its leader. That is, for any vertex $v$ with leader $u$, every edge $\{v, w\}$ will be replaced with an edge $\{u, w\}$ and all vertices in set $C(v)$ will be removed and added to set $C(u)$. We then remove vertex $v$ from the graph.
		\item Remove duplicate edges or self-loops and remove saturated/leader flags from the vertices.
	\end{enumerate}
\end{titledtbox}

Within the $\addtwohop$ subroutine, every active vertex $v$ attempts to connect itself to a subset of the vertices in its 2-hop. If there are more candidates than the budget of $v$ allows, we discard some of them arbitrarily. To formalize this, we use $N(u)$ to denote the neighbors of a vertex $u$.

Next, in the $\relabelinter$ subroutine, every vertex $v$ that sees a vertex $u$ of a higher level in its neighborhood, is ``relabeled'' to that vertex. That is, any occurrence of $v$ in the edges is replaced with $u$. As a technical point, it might happen that we end up with a chain $v_1 \to v_2 \to v_3 \to \ldots$ of relabelings where vertex $v_1$ has to be relabeled to $v_2$, $v_2$ has to be relabeled to $v_3$, and so on. In each iteration of the algorithm, we \emph{only apply the direct relabeling of every vertex}, that is $v_1$ ends up with label $v_2$,  $v_2$ ends up with label $v_3$, etc. An example of this is illustrated in Figure~\ref{fig:chains}.

\begin{figure}[t]
  \centering
  \includegraphics[scale=0.75]{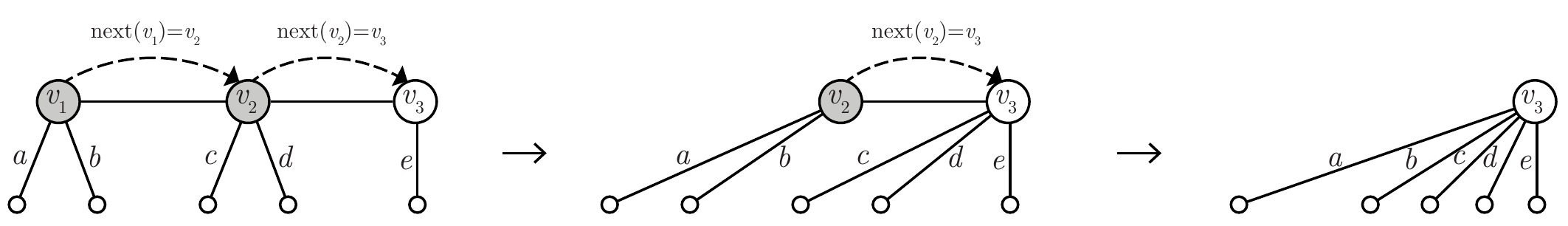}
  \vspace{-0.2cm}
  \caption{Algorithm~\ref{alg:connectivity} does not traverse ``relabeling chains''. In the first iteration, vertex $v_1$ is relabeled to $v_2$ and $v_2$ is relabeled to $v_3$.
After two iterations, both $v_1$ and $v_2$ are contracted to $v_3$. Note that the edge $\{v_2, v_3\}$ of iteration 1 will become a self-loop after applying the relabeling $v_2 \to v_3$ and thus will be removed. However, the edge $\{v_2, v_3\}$ that still remains in the second iteration is the result of applying relabelings $v_1 \to v_2$ and $v_2 \to v_3$ on edge $\{v_1, v_2\}$.}
  \label{fig:chains}
\end{figure}

Finally, the last subroutine $\relabelintra$, is where we increase the budgets/levels.

\subsection{Analysis of Algorithm~\ref{alg:connectivity} -- Correctness}\label{sec:alganalysis}

\paragraph{Correctness.} We first show that the algorithm indeed computes the connected components of the given graph.
The following lemma follows directly from the fact that Algorithm~\ref{alg:connectivity} does not split or merge connected components.

\begin{lemma}
	Let $\mathcal{S}_1, \ldots, \mathcal{S}_k$ be the connected components of $G$ at the end of Algorithm~\ref{alg:connectivity}.
Then, the family of sets $\{\cup_{v \in \mathcal{S}_i} C(v) \mid i \in [k]\}$ is equal to the family of vertex sets of connected components of the original graph.
\end{lemma}

\begin{proof}
We use induction to show that the claim is true at the end of each iteration $r$
 of Algorithm~\ref{alg:connectivity}.
Before we start the algorithm, i.e., when $r=0$, for every vertex $v$ we have $C(v) = \{v\}$. Therefore, clearly the base case holds. For the rest of the proof, it suffices to show that the three steps of $\addtwohop$, $\relabelinter$, and $\relabelintra$ maintain this property.
	
	Within the $\addtwohop$ subroutine, we only add edges to the graph. The only way that this operation may hurt our desired property, is if the added edges connect two different connected components of the previous iteration. However, every added edge is between two vertices of distance at most 2 (and thus in the same component) implying that this cannot happen.
	
	For the $\relabelinter$ subroutine, we first have to argue that the relabelings do not change the connectivity structure of the graph. It is clear that two disconnected components cannot become connected since each vertex is relabeled to another vertex of the same connected component. Moreover, we have to argue that one connected component does not become disconnected. For this, consider a path between two vertices $u$ and $v$ of the same component. After relabeling vertices, there is still a walk between the corresponding vertices to $u$ and $v$, thus they remain connected. Finally, observe that once a vertex $v$ is relabeled to some vertex $u$, in Line~\ref{line:passc} of the $\relabelinter$ subroutine, we add every vertex of $C(v)$ to $C(u)$. This ensures that for every component $\mathcal{S}$, the set $\cup_{v \in \mathcal{S}} C(v)$ does not lose any vertex and thus remains unchanged.
	
	Similarly, in the $\relabelintra$ step, the vertices only get contracted to the leaders in their 2-hop and once removed from the graph, a vertex $v$  passes every element in $C(v)$ to $C(u)$ of another vertex $u$ in its component, thus the property is maintained, concluding the proof.
\end{proof}

\subsection{Analysis of Algorithm~\ref{alg:connectivity} -- Round Complexity}\label{sec:round}

In order to pave the way for future discussions, we start with some definitions. We use $G_r=(V_r, E_r)$ to denote the resulting graph by the end of iteration $r$ of Algorithm~\ref{alg:connectivity}. Therefore, we have $V = V_0 \supseteq V_1 \supseteq V_2 \supseteq \ldots$ as we do not add any vertices to the graph. Moreover, for any vertex $v \in V$ and any iteration $r \geq 0$, we define $\next_r(v)$ to be the vertex $w \in V_r$ such that $v \in C(w)$ by the end of iteration $r$. That is, $\next_r(v)$ is the vertex that {\em corresponds} to $v$ by the end of iteration $r$. 

\begin{observation}
	Let $v \in V_r$ be an active vertex. Then for any $r' \leq r$, we have $\next_{r'}(v) = v$.
\end{observation}

For any iteration $r \geq 0$ and any vertex $v \in V$, we respectively use $\ell_r(v)$, $b_r(v)$ and $C_r(v)$ to denote the value of $\ell(\next_r(v))$, $b(\next_r(v))$ and $C(\next_r(v))$ by the end of iteration $r$. Furthermore, for any two vertices $v, u \in V$, we use $\dist_{r}(u, v)$ to denote the length of the shortest path between $\next_r(u)$ and $\next_r(v)$ in graph $G_r$.

The following claim implies that the corresponding level of a vertex is non-decreasing over time.

\begin{claim}\label{cl:levelnondecreasing}
	For any vertex $v \in V$ and any $r \geq r'$, we have $\ell_{r}(v) \geq \ell_{r'}(v)$.
\end{claim}
\begin{proof}
	We use induction on $r$. For the base case with $r = r'$, we clearly have $\ell_r(v) = \ell_{r'}(v)$. Suppose, by the induction hypothesis, that $\ell_{r-1}(v) \geq \ell_{r'}(v)$. If in iteration $r$ of the algorithm, vertex $\next_{r-1}(v)$ is not relabeled, i.e., if we have $\next_r(v) = \next_{r-1}(v)$, then $\ell_r(v) = \ell_{r-1}(v)$ by definition and the fact that the level of a particular vertex cannot decrease in one iteration. Therefore, by the induction hypothesis, we have $\ell_r(v) = \ell_{r-1}(v) \geq \ell_{r'}(v)$. On the other hand, if vertex $v$ is relabeled in iteration $r$, i.e., if $\next_r(v) \not = \next_{r-1}(v)$, then it suffices to show that it is relabeled to a vertex whose level is higher. This is clear from description of the algorithm. A vertex that gets relabeled within the $\relabelinter$ subroutine, does so if and only if the new vertex has a higher level. Similarly, in within the $\relabelintra$ subroutine of Algorithm~\ref{alg:connectivity}, any vertex $v$ that is contracted to another vertex does so if it is a marked saturated vertex of the same level, whose level increases by the end of the iteration.
\end{proof}

The next claim shows, in a similar way, that the distance between the corresponding vertices of two vertices $u$ and $v$ is non-increasing over time.

\begin{claim}\label{cl:distnonincreasing}
	For any two vertices $v, u \in V$ and any $r \geq r'$, we have $\dist_{r}(u, v) \leq \dist_{r'}(u, v)$.
\end{claim}
\begin{proof}
	Similar to the proof of Claim~\ref{cl:levelnondecreasing}, we can show this by induction on $r$ and, thus, the problem reduces to showing that in one iteration the corresponding distance between two vertices cannot increase. To show this, fix a shortest path $p$ between two vertices $v$ and $u$ at any iteration. Within the next iteration, the $\addtwohop$ subroutine does not affect this path as it only adds some edges to the graph. Moreover, the only effect of the relabeling steps on this path is that it may shrink it as one vertex of the path can be relabeled to one of its neighbors in the path. However, relabeling can in no way destroy or increase the length of this path. Thus, the lemma follows.
\end{proof}

Our next observation follows directly from the description of the algorithm.

\begin{observation}\label{obs:levellargerthanneighbors}
	For any $r \geq 0$ and any vertices $u, v \in V$ with $\dist_r(u, v) \leq 1$, we have $\ell_{r+1}(u) \geq \ell_r(v)$ and $\ell_{r+1}(v) \geq \ell_r(u)$.
\end{observation}
\begin{proof}
	This comes from the fact that any vertex who sees a neighbor of a higher level, is relabeled to its neighbor with the highest level in subroutine $\relabelinter$ of Algorithm~\ref{alg:connectivity}.
\end{proof}

\begin{claim}\label{cl:saturated}
	With high probability for any iteration $r$ and any vertex $v \in V_r$, if $v$ becomes saturated in the next iteration $r+1$, then there is at least one leader of the same level in its 2-hop, thus $\ell_{r+1}(v) \geq \ell_r(v)+1$.
\end{claim}
\begin{proof}
	If $v$ is saturated, then by definition, it has at least $b(v)$ vertices in its inclusive 2-hop  (i.e., the set $v \cup N(v) \cup N(N(v))$) that have the same level as that of $v$ and are also saturated. To see this, note that if $v$ is marked as saturated in Line~\ref{line:saturated1} of $\relabelintra$, then it has at least $b(v)$ other active direct neighbors with level at least $b(v)$ all of which will be marked as saturated in Line~\ref{line:saturated2}. Furthermore, if a vertex $v$ is marked as saturated in Line~\ref{line:saturated2}, then it has a saturated neighbor which has $b(v)$ direct saturated neighbors as just described. Thus $v$'s 2-hop will include $b(v)$ saturated vertices as desired.
	
	It suffices to show that one of these $b(v)$ saturated vertices will be marked as a leader with high probability. Recall that we mark each vertex independently with probability $\frac{3\log n}{b(v)}$, thus 
	\begin{equation*}
		\Pr\Big[\ell_{r+1}(v) = \ell_r(v) \Big] \leq \Big(1-\frac{3\log n}{b(v)}\Big)^{b(v)} \leq \exp(-3\log n) \leq 1/n^{3}.
	\end{equation*}
	By a union bound over all vertices, and over the total number of iterations of the algorithm which is clearly less than $n^2$, we get that with probability at least $1-1/n$ every vertex that gets saturated sees a marked vertex in its 2-hop and its corresponding level will thus be increased in the next iteration.
\end{proof}
	
The next lemma highlights a key property of the algorithm and will be our main tool in analyzing the round complexity. Intuitively, it shows that with high probability, after every 4 iterations of Algorithm~\ref{alg:connectivity}, every vertex $v$ is either connected to its 2-hop, or its corresponding level increases by at least 1.
	
	\begin{lemma}\label{lem:2hoporlevelincrease}
		Let $u, v \in V$ be two vertices with $\dist_r(u, v) = 2$ for some iteration $r$. If $\dist_{r+4}(u, v) \geq 2$, then $\ell_{r+4}(v) \geq \ell_r(v) + 1$ and $\ell_{r+4}(u) \geq \ell_r(u) + 1$. This holds for all vertices $u$ and $v$ and over all iterations of the algorithm with high probability.
	\end{lemma}
	
	\begin{proof}
		By Claim~\ref{cl:distnonincreasing}, we have $\dist_{r+4}(u, v) \leq  \dist_r(u, v) = 2$. As such, to prove the lemma, it suffices to obtain a contradiction by assuming that $\dist_{r+4}(u, v) = 2$ and (w.l.o.g.) $\ell_{r+4}(u) = \ell_{r}(u)$.
		
		Recall that the lemma assumes that $\dist_r(u, v)=2$. Therefore, there must exist a vertex $w$ with $\dist_r(u, w)=\dist_r(w, v)=1$. By an application of Observation~\ref{obs:levellargerthanneighbors}, we have $$\ell_{r+4}(u) \geq \ell_{r+3}(w) \geq \ell_{r+2}(v) \geq \ell_{r+1}(w) \geq \ell_r(u).$$
		Combining this with our assumption that $\ell_{r+4}(u) = \ell_r(u)$, we get
		\begin{equation}\label{eq:03813}
			\ell_{r+4}(u) = \ell_{r+3}(w) = \ell_{r+2}(v) = \ell_{r+1}(w) = \ell_r(u).
		\end{equation}
		Moreover, by Claim~\ref{cl:levelnondecreasing} which states the levels are non-decreasing over time, we have 
		\begin{equation}\label{eq:37249}
			\ell_{r+4}(u) \geq \ell_{r+2}(u) \geq \ell_r(u) \qquad \text{and} \qquad \ell_{r+3}(w) \geq \ell_{r+2}(w) \geq \ell_{r+1}(w).
		\end{equation}
		Combination of (\ref{eq:03813}) and (\ref{eq:37249}) directly implies the following two useful inequalities.
		
		\begin{observation}\label{obs:equalatround2}
			 $\ell_{r+2}(u) = \ell_{r+2}(w) = \ell_{r+2}(v)$.
		\end{observation}
		\begin{proof}
			Inequality $\ell_{r+4}(u) \geq \ell_{r+2}(u) \geq \ell_r(u)$ of (\ref{eq:37249}) combined with equality $\ell_{r+4}(u) = \ell_r(u)$ of (\ref{eq:03813}) implies $\ell_{r+2}(u) = \ell_r(u)$. This combined with $\ell_{r}(u) = \ell_{r+2}(v)$ of (\ref{eq:03813}), gives $\ell_{r+2}(u) = \ell_{r+2}(v)$. 
			
			Inequality $\ell_{r+3}(w) \geq \ell_{r+2}(w) \geq \ell_{r+1}(w)$ of (\ref{eq:37249}) combined with equality $\ell_{r+3}(w)=\ell_{r+1}(w)$ of (\ref{eq:03813}) implies $\ell_{r+2}(w) = \ell_{r+1}(w)$. Combined with $\ell_{r+1}(w) = \ell_{r+2}(v)$ of (\ref{eq:03813}), it gives $\ell_{r+2}(w) = \ell_{r+2}(v)$. 
		\end{proof}
		
		\begin{observation}\label{obs:contradictionequality}
			$\ell_{r+4}(u) = \ell_{r+2}(u)$.
		\end{observation}
		\begin{proof}
		Inequality $\ell_{r+4}(u) = \ell_r(u)$ due to (\ref{eq:03813}) combined with inequality $\ell_{r+4}(u) \geq \ell_{r+2}(u) \geq \ell_r(u)$ of (\ref{eq:37249}) implies $\ell_{r+4}(u) = \ell_{r+2}(u) = \ell_r(u)$.
		\end{proof}
		
		Observation~\ref{obs:equalatround2} implies that the corresponding levels of all three vertices $u$, $w$ and $v$ should be the same at the end of iteration $r+2$. Thus, within the $\addtwohop$ subroutine of iteration $r+3$, we have $\next_{r+2}(v) \in H(\next_{r+2}(u))$; now either we connect $\next_{r+2}(u)$ and $\next_{r+2}(v)$ which reduces their distance to 1 contradicting our assumption that $\dist_{r+4}(u, v) = 2$, or otherwise vertex $\next_{r+2}(u)$ spends its budget to get connected to at least $b_{r+2}(u)$ other vertices of level at least $\ell_{r+2}(u)$. Let $N$ be the set of these neighbors of $\next_{r+2}(u)$. There are three scenarios and each leads to a contradiction:
		\begin{itemize}
			\item If for a vertex $ x \in N$, $\ell_{r+2}(x) > \ell_{r+2}(u)$, then by Observation~\ref{obs:levellargerthanneighbors}, the level of the corresponding vertex of $u$ increases in the next iteration and we have $\ell_{r+4}(u) \geq \ell_{r+3}(u) > \ell_{r+2}(u)$ which contradicts equality $\ell_{r+4}(u) = \ell_{r+2}(u)$ of Observation~\ref{obs:contradictionequality}.
			\item If a vertex $x \in N$ is inactive, then $x$ is in a chain by definition of inactive vertices. Every vertex in a chain has a vertex of higher level next to it, thus $\ell_{r+3}(x) > \ell_{r+2}(x)$ by Observation~\ref{obs:levellargerthanneighbors}. Furthermore, since $x \in N$, we know $\ell_{r+2}(x) \geq \ell_{r+2}(u)$. This means that $\next_{r+3}(u)$ has a neighbor of strictly higher level, thus by Observation~\ref{obs:levellargerthanneighbors}, we have to have $\ell_{r+4}(u) > \ell_{r+2}(u)$ which contradicts equality $\ell_{r+4}(u) = \ell_{r+2}(u)$ of Observation~\ref{obs:contradictionequality}.
			\item  If the two cases above do not hold, then after applying $\addtwohop$ in iteration $r+2$, $\next_{r+2}(u)$ has at least $b_{r+2}(u)$ active neighbors of level exactly $\ell_{r+2}(u)$. Furthermore, vertex $\next_{r+2}(u)$ itself has to be active, or otherwise its corresponding level has to increase in the next iteration which is a contradiction. This means by definition that $\next_{r+2}(u)$ is saturated during iteration 3. By Claim~\ref{cl:saturated}, with high probability the corresponding level of every saturated vertex increases by at least one in the next iteration, and thus we get $\ell_{r+3}(u) > \ell_{r+2}(u)$ which, again, would imply $\ell_{r+4}(u) \not= \ell_{r+2}(u)$ contradicting Observation~\ref{obs:contradictionequality}.
		\end{itemize}
		To wrap up, we showed that if the distance between the corresponding vertices to $u$ and $v$ after the next 4 iterations is not decreased to at most 1, then the corresponding level of $u$ and $v$ has to go up by one with high probability.
	\end{proof}
	
As discussed before, Lemma~\ref{lem:2hoporlevelincrease} implies that after every $O(1)$ consecutive iterations of Algorithm~\ref{alg:connectivity}, each vertex either is (roughly speaking) connected to the vertices in its 2-hop or sees a level increase. It is easy to show that if {\em every} vertex is connected to the vertices in its 2-hop, the diameter of the graph is reduced by a constant factor, and thus after $O(\log D)$ iterations every connected component becomes a clique. Notice, however, that Lemma~\ref{lem:2hoporlevelincrease} does not guarantee this, as for some vertices, we may only have a level increase instead of connecting them to their 2-hop. Let $L$ be an upper bound on the level of the vertices throughout the algorithm. (We later show in Lemma~\ref{lem:levelsnomorethanloglog} that $L = O(\log \log_{T/n}n)$.) Since the maximum possible level is $L$, each vertex does not connect 2-hops for at most $L$ iterations. Therefore, if for instance, within each of the first $L$ iterations of the algorithm, the corresponding level of {\em every} vertex increases, we cannot have any level-increases afterwards. Therefore within the next $O(\log D)$ iterations, each vertex connects 2-hops and every connected component becomes a clique. Overall, this takes $O(L + \log D)$ iterations. In reality, however, the level increases do not necessarily occur in bulk within the first $L$ iterations of the algorithm. In fact, Lemma~\ref{lem:2hoporlevelincrease} alone is not enough to show a guarantee of $O(L + \log D)$. To get around this problem, we need to use another crucial property of the algorithm highlighted in Observation~\ref{obs:levellargerthanneighbors}. A proof of sketch of how we combine these two properties to get our desired bound was already given in Section~\ref{sec:highlevel}. The following lemma formalizes this.

\begin{lemma}\label{lem:lplusloglog}
	Let $L$ be an upper bound on the number of times that the corresponding level of a vertex may increase throughout the algorithm. Only $O(L + \log D)$ iterations of the for loop in Algorithm~\ref{alg:connectivity} suffices to make sure that with high probability, every remaining connected component becomes a clique.
\end{lemma}
\begin{proof}
	It will be convenient for the analysis to call every 4 consecutive iterations of the for-loop in Algorithm~\ref{alg:connectivity}   a {\em super-iteration}. That is, for any $i\geq 1$, we define the $i$th super-iteration to be the combination of performing iterations $4i-3, 4i-2, \ldots, 4i$ of Algorithm~\ref{alg:connectivity}.

	Fix two arbitrary vertices $u$ and $v$ in a connected component of the original graph $G$. It suffices to show that after running the algorithm for $R := O(L+\log D)$ super-iterations, the corresponding vertices to $u$ and $v$, are at distance at most 1. To show this, we maintain a path between $u$ and $v$ and update it over time. We use $P_r$ to denote the maintained path by the end of super-iteration $r$, i.e., the path is updated every four iterations. The initial path, $P_0$, is any arbitrary shortest path between $u$ and $v$ in the original graph $G$; notice that $P_0$ has at most $D+1$ vertices, as the diameter of $G$ is $D$. As we move forward, $u$ and $v$ may be relabeled; nonetheless, the path $P_r$ will be a path from vertex $\next_{4r}(u)$ (which is the corresponding vertex to $u$ by the end of iteration $4r$ or equivalently super-iteration $r$) to vertex $\next_{4r}(v)$. Crucially, the path $P_r$ is not necessarily the shortest path between $\next_{4r}(v)$ and $\next_{4r}(u)$ in $G_r$. The reason is that the naive shortest paths may ``radically'' change from one iteration to another. Instead, we carefully construct $P_r$ to ensure that it passes only through the  corresponding vertices of the vertices in $P_{r-1}$, which also inductively indicates that every vertex in $P_r$ is in set $\{\next_{4r}(w) \, | \, w \in P_0 \}$.

To use these gradual updates, for every  $r$, we define a potential function $\Phi_r : V(P_r) \to \mathbb{N}$ that maps every vertex of path $P_r$ to a positive integer. The definition of function $\Phi_r$ and construction of path $P_r$ are recursively based on $\Phi_{r-1}$ and $P_{r-1}$. As for the base case, we have $\Phi_0(v) = 1$ for every vertex $v \in P_0$. For the rest of the iterations, we follow the following steps.

To construct $P_r$ from $P_{r-1}$, we first apply the relabelings of iterations $4r-3, \ldots, 4r$, on the vertices in $P_{r-1}$. That is, the sequence $P_{r-1} = (w_1, \ldots, w_s)$ becomes $Q = (q_1, \ldots, q_s)$ where $q_i = \next_{4r}(w_i)$. Note that multiple vertices in $P_{r-1}$ may have been relabeled to the same vertex throughout these four iterations, and thus the elements in $Q$ are not necessarily unique. Next, we use an $s$-bit {\em mask} vector $K \in \{0, 1\}^s$ to denote a subsequence\footnote{A {\em subsequence} is a derived from another sequence by deleting some or no elements of it without changing the order of the remaining elements.} of $Q$ that corresponds to the vertices in $P_r$. That is, $P_r$ contains the $i$th element of $Q$ if and only if $K_i=1$. To guarantee that $P_r$ is indeed a path and that it has some other useful properties, our mask vector $K$ should satisfy the following properties:
	\begin{enumerate}[label={(P\arabic*)}, ref=P\arabic*]
		\item \label{prop:1} $K_1 = K_s = 1$.
		\item \label{prop:2} If for some $i, j \in [s]$ with $i \not = j$, we have $q_i = q_j$, then at most one of $K_i$ and $K_j$ is 1.
		\item \label{prop:3} If for some $1 \leq i < j \leq s$ with $K_i = K_j = 1$, there is no $k$ with $i < k < j$ for which $K_k = 1$, then $q_i$ and $q_j$ should have a direct edge in graph $G_{4r}$.
		\item \label{prop:4} If for some $i \in [s]$, we have $\ell_{4r}(q_i) = \ell_{4(r-1)}(w_i)$ (i.e., the level of the corresponding vertex to $w_i$ is not increased) and $K_i = 1$, then at least one of $K_{i-2}$, $K_{i-1}$, $K_{i+1}$ or $K_{i+1}$ should be 0.\footnote{$K$ is an $s$-bit vector, but assume for preciseness of definition that $K_{-1} = K_0 = K_{s+1} = K_{s+2} = 1$.}
	\end{enumerate}
Property~\ref{prop:1} guarantees that the path of the next iteration remains to be between the corresponding vertices to $u$ and $v$. Property~\ref{prop:2} ensures that we do not revisit any vertex in $P_r$ which is necessary if we want $P_r$ to be a path. Property~\ref{prop:3} ensures that every two consecutive vertices in $P_r$ are neighbors in $G_{4r}$, which again, is necessary if we want $P_r$ to denote a path in $G_{4r}$. Finally, Property~\ref{prop:4} guarantees that if the corresponding level of a vertex $w_i \in V_{4(r-1)}$ does not increase in iterations $4r-3, \ldots, 4r$, and that $q_i$ (which is the corresponding vertex to $w_i$ after these four iterations) is included in path $P_r$, there is a vertex in $\{w_{i-2}, w_{i-1}, w_{i+1}, w_{i+2}\}$ whose corresponding vertex at the next iteration is not included in $P_r$. Note that we have to be careful that by satisfying Property~\ref{prop:4}, we do not violate Property~\ref{prop:3}. In other words, we have to make sure that once we drop the (2-hop) neighboring vertices of $w_i$ in $P_{r-1}$ from $P_r$, $P_r$ remains to be a connected path. However, this can be guaranteed by Lemma~\ref{lem:2hoporlevelincrease} which says if the corresponding level of a vertex does not increase in 4 iterations, its distance to the vertices in its 2-hop decreases to at most 1 (see also Section~\ref{sec:highlevel}). Overall, we get the following result.

\begin{claim}
	If $q_1 \not = q_s$ and if $|P_{r-1}| > 3$, then with high probability there exists a mask vector $K$ satisfying Properties~\ref{prop:1}-\ref{prop:4}.
\end{claim}

Construction of function $\Phi_r$ is also based on the mask vector $K$ that we construct $P_r$ with. Recall that $\Phi_r$ is a function from the vertices in $P_r$ to $\mathbb{N}$. Therefore, in order to describe $\Phi_r$, it suffices to define the value of $\Phi_r$ on vertex $q_i$ iff $K_i = 1$. Assuming that $K_i = 1$, define $l_i$ to be the smallest number in $[1, i]$ such that $\sum_{j \in l_i}^{i-1} K_j \leq 1$. In a similar way, define $r_i$ to be the largest number in $[i, s]$ where $\sum_{j = i+1}^{r_i} K_j \leq 1$. Having these, we define $\Phi_r(q_i)$ in the following way:
\begin{equation}\label{eq:defphi}
	\Phi_r(q_i) := \Phi_{r-1}(w_i) + \frac{1}{4}\sum_{j=l_i}^{i-1} \Phi_{r-1}(w_j)  + \frac{1}{4}\sum_{j=i+1}^{r_i} \Phi_{r-1}(w_j).
\end{equation}

The next two claims are the main properties of function $\Phi_r$ that we use in proving Lemma~\ref{lem:lplusloglog}.

\begin{claim}\label{cl:potentialincreasing}
	For any $r \geq 0$ and any vertex $w \in P_r$ with level $\ell = \ell_{4r}(w)$, we have $\Phi_r(w) \geq (1.25)^{r - \ell}$.
\end{claim}
\begin{proof}
	We use induction on $r$. For the base case with $r=0$, we have $\Phi_0(w) = 1$ and since it is before the first iteration, we have $\ell = 0$. Thus, we have $\Phi_0(w) \geq (1.25)^{0 - 0} = 1$. The induction hypothesis guarantees for every vertex $w'$ of path $P_{r-1}$ with level $\ell'$, that $\Phi_{r-1}(w') \geq (1.25)^{r-1-\ell'}$. We show that this carries over to the vertices of $P_r$ as well.
	
	We would like to prove that for every vertex $w \in P_r$, we have $\Phi_r(q) \geq (1.25)^{r-\ell_{4r}(q)}$. We know by construction of $P_r$ from $P_{r-1}$ that vertex $w$ of $P_r$ is the corresponding vertex of some vertex $w'_i \in P_{r-1}$ with $K_i = 1$ where $K$ denotes the mask vector that we use to construct $P_r$ from $P_{r-1}$, i.e., $w = \next_{4(r-1)}^4(w'_i)$. By the induction hypothesis, we have 
	\begin{equation}\label{eq:73624}
	\Phi_{r-1}(w'_i) \geq (1.25)^{r-1-\ell_{4(r-1)}(w'_i)}.	
	\end{equation}
	Therefore, if during super-iteration $r$, the corresponding level of $w_i$ increases, i.e., if we have $\ell_{4r}(w) \geq \ell_{4(r-1)}(w'_i) + 1$, then we have
	\begin{equation*}
		\Phi_r(w) \geq \Phi_{r-1}(w'_i) \stackrel{\text{By (\ref{eq:73624})}}{\geq} (1.25)^{r - 1 - \ell_{4(r-1)}(w'_i)} \geq (1.25)^{r - \ell_{4r}(w)},
	\end{equation*}
	where the first inequality comes from the fact that $\Phi_r(w_i) > \Phi_{r-1}(w'_i)$ which itself is directly followed by (\ref{eq:defphi}). This means that if the corresponding level of $w'_i$ remains unchanged within super-iteration $r$, we have our desired bound on $\Phi_r(w)$. The only scenario that is left is if the corresponding level of $w'_i$ remains unchanged, i.e., $\ell_{4r}(w) = \ell_{4(r-1)}(w'_i)$.
	
	If the corresponding level of $w'_i$ remains unchanged during super-iteration $r$, then by Property~\ref{prop:4} of the mask vector $K$, at least one of $\{K_{i-2}, K_{i-1}, K_{i+1}, K_{i+2}\}$ is 0. Suppose without loss of generality that $K_{i-2} = 0$. First, observe that we have to have
		\begin{equation}\label{eq:81234}
			\ell_{4(r-1)}(w'_{i-2}) \leq \ell_{4(r-1)}(w'_i).
		\end{equation}
	The reason is that if level of $w'_{i-2}$, which has distance at most 2 from $w'_{i}$ in graph $G_{4(r-1)}$, has a higher level than $w'_i$, then by Observation~\ref{obs:levellargerthanneighbors}, the corresponding level of $w'_i$ after at most 2 iterations (still within a super-iteration) should increase to at least $\ell_{4(r-1)}(w'_{i-2})$ which would contradict the assumption that the corresponding level of $w'_i$ remains unchanged for 4 iterations. This means that by the induction hypothesis, now on vertex $w'_{i-2}$ of path $P_{r-1}$, we have
	\begin{equation}\label{eq:62128}
		\Phi_{r-1}(w'_{i-2}) \geq (1.25)^{r-1-\ell_{4(r-1)}(w'_{i-2})} \stackrel{\text{By (\ref{eq:81234})}}{\geq} (1.25)^{r-1-\ell_{4(r-1)}(w'_i)}.
	\end{equation}
	Now, recall that we assumed $K_{i-2} = 0$. This means, by construction of $\Phi_r$ using (\ref{eq:defphi}), that we have to have
	\begin{equation}\label{eq:71239}
		\Phi_r(w) \geq \Phi_{r-1}(w'_i) + \frac{1}{4} \Phi_{r-1}(w'_{i-2}).
	\end{equation}
	Therefore, we have
	\begin{align*}
 		\Phi_r(w) &\geq \Phi_{r-1}(w'_i) + \frac{1}{4} \Phi_{r-1}(w'_{i-2}) && \text{By (\ref{eq:71239}).}\\
 			&\geq \Big((1.25)^{r-1-\ell_{4(r-1)}(w'_i)}\Big) + \frac{1}{4}\Big((1.25)^{r-1-\ell_{4(r-1)}(w'_i)}\Big) && \text{By induction hypothesis and (\ref{eq:62128}).}\\
 			&= 1.25 \Big((1.25)^{r-1-\ell_{4(r-1)}(w'_i)}\Big) = (1.25)^{r-\ell_{4(r-1)}(w'_i)}\\
 			&\geq (1.25)^{r-\ell_{4r}(w)}. && \text{Since $\ell_{4r}(w) \geq \ell_{4(r-1)}(w'_i)$.}
	\end{align*}
	Concluding the proof of Claim~\ref{cl:potentialincreasing}. 
\end{proof}

\begin{claim}\label{cl:potentialsumconstant}
	For any $r\geq 0$ with $|P_r| > 3$, we have $\sum_{w \in P_r}\Phi_r(w) \leq D+1$.
\end{claim}
\begin{proof}
	The inequality $\sum_{v\in P_0} \Phi_r(v) \leq D + 1$ is followed by the fact that $P_0$, which is a shortest path between $u$ and $v$ in the original graph has at most $D+1$ vertices and that $\Phi_0(w) = 1$ for any vertex $w \in P_0$. Moreover, one can easily show that for any $r > 0$, we have $\sum_{w \in P_r}\Phi_r(w) \leq \sum_{w \in P_{r-1}}\Phi_{r-1}(w)$ directly by the definition of $\Phi_r$ from $\Phi_{r-1}$ and Property~\ref{prop:1} of the mask vectors used. Combining these two facts via a simple induction on $r$ proves the claim.
\end{proof}

We are now ready to prove Lemma~\ref{lem:lplusloglog}. Run the algorithm for $R := L + 4\log D$ super-iterations. If path $P_R$ has at most 3 vertices, we are done since our goal is to show that the distance between the corresponding vertices of $u$ and $v$ in graph $G_R$ is at most 2 --- which itself would imply that every connected component in $G_R$ has diameter at most 2. In fact, we show that this should always be the case.\footnote{More precisely, the ``always'' here is conditional on the assumption that our high probability events hold. This is not a problem since otherwise we say the algorithm fails and this happens with probability at most $1/n$.} Suppose for the sake of contradiction that we can continue to super-iteration $R$ in constructing $P_R$  and $\Phi_R$ and still have $|P_R| > 3$. Let $u_R := \next_{4R}(u)$ be the corresponding vertex to vertex $u$ by the end of super-iteration $R$. Property~\ref{prop:1} of our mask vectors in constructing paths $P_1, \ldots, P_R$ ensures that path $P_R$ should start with vertex $u_R$. By Claim~\ref{cl:potentialincreasing}, we have 
\begin{equation}\label{eq:19294}
	\Phi_R(u_R) \geq (1.25)^{R - \ell_R(u_R)} \geq (1.25)^{R - L},
\end{equation}
where the latter inequality comes from the assumption that $L$ is an upper bound on the level of every vertex. Now, since $R = L + 4\log D$, we have 
\begin{equation}\label{eq:91239}
R - L \geq 4 \log D.
\end{equation}
Combining (\ref{eq:19294}) with (\ref{eq:91239}) we get
\begin{equation*}
	\Phi_R(u_R) \geq (1.25)^{4\log D} > D + 1.
\end{equation*}
However, this contradicts with Claim~\ref{cl:potentialsumconstant} which guarantees $\Phi_R(u_R)$ should be less than $D+1$. Therefore, our initial assumption that $R$ can be as large as $L + 2\log D$ cannot hold; meaning that in $O(L + \log D)$ iterations, the remaining graph will be a collection of connected components of diameter $O(1)$.

Once the diameter of every remaining connected component gets below $O(1)$, it is easy to confirm that in the next $O(L)$ iterations of the algorithm, the diameter reduces to 1 (i.e., every connected component becomes a clique). To see this, note that since the diameter is $O(1)$, the maximum level within each component propagates to all the vertices in $O(1)$ iterations. If this budget is not enough for a vertex to connect 2-hop, its level increases by Lemma~\ref{lem:2hoporlevelincrease}. This level, again, propagates to all other vertices. Eventually, after the next $O(L)$ iterations, the vertices will reach the maximum possible level and thus have enough budget to get connected to every remaining vertex in the component.

Overall, it takes $O(L + \log D)$ iterations until the diameter of every remaining connected component becomes $O(1)$ and after that, at most $O(L)$ other iterations for them to become cliques.
\end{proof}

To continue, we give the following upper bound on the levels.

\begin{lemma}\label{lem:levelsnomorethanloglog}
	For any vertex $v$, the value of $\ell(v)$ never exceeds $O(\log\log_{T/n} n)$.
\end{lemma}
\begin{proof}
	Observe that the only place throughout Algorithm~\ref{alg:connectivity} that we increase the level of a vertex is in Line~\ref{line:levelincrease} of the $\relabelintra$ procedure. Within this line, the budget of the vertex is also increased from $b(v)$ to $b(v)^{\budgetconstant}$. Now, given that the initial budget of every vertex is $\beta_0 = (T/n)^{1/2}$, throughout the algorithm, we have $b(v) = \beta_0^{\budgetconstant^{\ell(v)}}$. On the other hand, observe that if a vertex reaches a budget of $n$, it will not be marked as saturated, and thus, we do not update its level/budget anymore. Therefore, we have
	$b(v) = \beta_0^{(\budgetconstant)^{\ell(v)}} \leq n$ which means 
	$\ell(v) \leq \log_{\budgetconstant}\log_{\beta_0} n = O(\log\log_{T/n} n)$.
\end{proof}

Combining the two lemmas above, we can prove the following bound on the round complexity.

\begin{lemma}\label{lem:rounds}
	With high probability the number of rounds executed by Algorithm~\ref{alg:connectivity} is $O(\log D + \log\log_{T/n} n)$.
\end{lemma}
\begin{proof}
	By Lemma~\ref{lem:lplusloglog}, it takes only $O(L + \log D)$ iterations for Algorithm~\ref{alg:connectivity} to halt where $L$ is an upper bound on the level of the vertices. Lemma~\ref{lem:levelsnomorethanloglog} shows that $L = O(\log\log_{T/n} n)$. Therefore, the round complexity of Algorithm~\ref{alg:connectivity} is $O(\log D + \log \log_{T/n} n)$.
\end{proof}

\subsection{Analysis of Algorithm~\ref{alg:connectivity} -- Implementation Details \& Space}\label{sec:space}

\begin{lemma}\label{lem:totalspace}
	The total space used by Algorithm~\ref{alg:connectivity} is $O(T)$.
\end{lemma}
\begin{proof}
To bound the total space used by Algorithm~\ref{alg:connectivity}, we have to bound the number of edges that may exist in the graph. More specifically, we have to show that within the \addtwohop{} subroutine, we do not add too many edges to the graph. Recall that we control this with the budgets. It is not hard to argue that sum of budgets of remaining vertices in each round of the algorithm does not exceed $T$. However, there is a subtle problem that prevents this property to be sufficient for bounding the number of edges in the graph. The reason is that throughout the algorithm, the degree of a vertex may be much larger than its budget. For instance in the first iteration, a vertex may have a degree of up to $\Omega(n)$ while the budgets are much smaller.

For the analysis, we require a few definitions. For every iteration $r$ and any vertex $v \in V_r$, we define $d_r(v)$ to be the number of neighbors of $v$ in $G_r$ with level at least $\ell_r(v)$. Moreover, we define the {\em remaining budget} $s_r(v)$ of $v$ to be $\max\{0, b_r(v) - d_r(v)\}$ if $v$ is active and 0 otherwise. To clarify the definition, note that within the \addtwohop{} subroutine, each vertex $v$ connects to at most $s_r(v)$ new vertices. We further define
$$
y_r := |E_r| + \sum_{v \in V_r} s_r
$$
to be the {\em potential space} by the end of iteration $r$. It is clear by definition that $y_r$ is an upper bound on the total number of edges in the graph after iteration $r$. Therefore it suffices to show that $y_r = O(T)$ for any $r$. The base case follows immediately:
\begin{observation}
	$y_0 = O(T)$.
\end{observation}
\begin{proof}
	We have $y_0 = |E_0| + \sum_{v \in V_0}s_r \leq m + n \cdot (\frac{T}{n})^{1/2} < m + T \leq O(T)$ where the last inequality comes from the fact that $T = \Omega(m)$.
\end{proof}
In what follows, we argue that for any $r$, we have $y_r \leq y_0 + O(T) = O(T) + O(T) = O(T)$ as desired. To do this, we consider the effect of each of the three subroutines of Algorithm~\ref{alg:connectivity} in any iteration $r+1$ on the value of $y_{r+1}$ compared to $y_r$. We first show that the two procedures $\addtwohop$ and $\relabelinter$ cannot increase the potential space. We then give an upper bound of $O(T)$ on the increase in the potential space due to procedure $\relabelintra$ over the course of the algorithm (i.e., not just one round). 

\smparagraph{$\addtwohop$ procedure.} In the \addtwohop{} procedure, each vertex $v$ connects itself to at most $s_r(v)$ other vertices of level at least $\ell(v)$ in its 2-hop as described above. These added edges, will then decrease the remaining budget of $v$ by definition. Therefore, for any edge that is added to the graph, the remaining budget of at least one vertex is decreased by 1. Thus, the total potential space cannot increase.

\smparagraph{$\relabelinter$ procedure.} Next, within the $\relabelinter$ procedure, we do not add any edges to the graph. Therefore, the only way that we may increase the potential space is by increasing the remaining budget of the vertices. If a vertex gets relabeled to a higher level neighbor, the algorithm marks it as inactive; this by definition decreases its remaining budget to 0. As such, it only suffices to consider the remaining budget of the vertices that are not relabeled; take one such vertex $v$. Recall that the remaining budget of $v$ depends on the level of the neighbors of $v$ as well. The crucial property here is that whenever a vertex is relabeled to a neighbor, its corresponding level is increased. This implies that the change in the corresponding level of $v$'s neighbors cannot increase the remaining budget of $v$.

There is still one way that $v$'s remaining budget may increase: if an edge $\{v, u\}$ with $\ell_r(u) \geq \ell_r(v)$ is removed from the graph. Recall that an edge may be removed from the graph within  Line~\ref{line:removeduplicateandselfloops} of $\relabelinter$ where we remove duplicate edges or self-loops. Note that if removal of an edge increases the remaining budget of \underline{one} of its endpoints only, then the potential space $y_{r+1}$ does not change as the increase in $\sum_{v} s_{r+1}(v)$ is canceled out by the decrease in $|E_{r+1}|$. However, we have to argue that removal of an edge cannot increase the remaining budget of its both end-points. To see this, observe that the graph, before the $\relabelinter$ procedure cannot have any duplicate edges or self-loops (as we must have removed them before) and all these edges have been created within this iteration. Take an edge $\{u, v\}$ and suppose that there are multiple duplicates of it. All, but at most one, of duplicates of $\{u, v\}$ are the result of the relabelings. Call these the {\em relabeled} edges and suppose due to symmetry that any removed edge is relabeled. Consider an edge $e'$ that is relabeled to $\{u, v\}$ and is then removed. At least one of endpoints of $e'$ must be some vertex $w$ which is relabeled to either $u$ or $v$, say $u$ w.l.o.g. An equivalent procedure is to remove $e'$ before $w$ is relabeled to $u$ and the outcome would be the same. Since $u \not\in e'$, there is no way that removing $e'$ would change the remaining budget of $u$. On the other hand, since $w$ is relabeled and does not survive to $V_{r+1}$ it does not have any effect on $s_{r+1}$. This means that removing any duplicate edge increases sum of remaining budgets by at most 1 thus the potential space cannot increase.

\smparagraph{\relabelintra{} procedure.} We showed that subroutines $\addtwohop$ and $\relabelinter$ cannot increase the potential space of the previous round. Here, we consider the effect of the last subroutine $\relabelintra$. Similar to $\relabelinter$, we do not add any edges to the graph. Therefore, we only have to analyze the remaining budgets after this procedure.

First take a vertex $v$ that is not marked as saturated. The remaining budget of $v$ may increase if some of its edges are removed because of duplicates which are caused by contracting saturated vertices to their leaders. However, precisely for the same argument that we had for the $\relabelinter$ procedure, removal of an edge can only increase the remaining budget of at most one of its end-points thus this does not increase the potential space.

Next, if a vertex $v$ is marked as saturated but is not marked as a leader, by Claim~\ref{cl:saturated} it is, w.h.p., going to get contracted to a leader and removed from the graph. Therefore, the only case for which the remaining budget of a vertex may increase is for saturated vertices that are marked as leaders. We assume the worst case. That is, we assume that if a vertex $v$ is saturated and is marked as a leader within iteration $r+1$, then the potential space is increased by its new budget $b_{r+1}(v)$ (note, by definition, that remaining budget can never be larger than budget). Instead of analyzing the effect of this increase within one iteration, we show that the total sum of such increases over all iterations of the algorithm is bounded by $O(T)$.

Let us use $\beta_i$ to denote the budget of vertices with level $i$ and use $n_i$ to denote the number of vertices that have been selected as a leader over the course of algorithm for at least $i$ times. In other words, $n_i$ denotes the total number of vertices that reach a level of at least $i$ throughout the algorithm. We can bound sum of increases in potential space due to the \relabelintra{} procedure over all iterations of the algorithm by:
\begin{equation}\label{eq:increaserb}
	\sum_{i=1}^{\infty} \beta_i \cdot n_i.
\end{equation}
Thus it suffices to bound this quantity by $O(T)$.

\begin{claim}\label{cl:brecurse}
	For any $i \geq 1$, we have $\beta_i = (\beta_{i-1})^{\budgetconstant}$ and have $\beta_0 = (T/n)^{1/2}$.
\end{claim}
\begin{proof}
	We have $\beta_0 = (T/n)^{1/2}$ for Line~\ref{line:initialize} of Algorithm~\ref{alg:connectivity}. Furthermore, we have $\beta_i = (\beta_{i-1})^{\budgetconstant}$  due to Line~\ref{line:levelincrease} of \relabelintra{} which is the only place we increase the level of a vertex and at the same time increase its budget from $b$ to $b^{\budgetconstant}$.
\end{proof}
We also have the following bound on $n_i$:

\begin{claim}\label{cl:nrecurse}
	For any $i \geq 1$ we have $n_i < n_{i-1} \cdot (\beta_{i-1})^{-0.25}$.
\end{claim}

Before describing the proof of Claim~\ref{cl:nrecurse}, let us first see we can get an upper bound of $O(T)$ for the value of (\ref{eq:increaserb}). For any $i \geq 1$, we have
\begin{equation}\label{eq:12834710239847}
\beta_i \cdot n_i 
\stackrel{\text{Claim~\ref{cl:brecurse}}}{=} 
(\beta_{i-1})^{\budgetconstant} \cdot n_i 
\stackrel{\text{Claim~\ref{cl:nrecurse}}}{<}
(\beta_{i-1})^{\budgetconstant} \cdot (\beta_{i-1})^{-0.25} \cdot n_{i-1}  = \beta_{i-1} \cdot n_{i-1}.
\end{equation}
On the other hand, recall by Lemma~\ref{lem:levelsnomorethanloglog} that the maximum possible level for a vertex is $L = O(\log \log_{T/n} n)$, meaning that for any $i > L$ we have $n_i = 0$; thus:
$$
	\sum_{i=1}^{\infty} \beta_i \cdot n_i = 	\sum_{i=1}^{L} \beta_i \cdot n_i
	\stackrel{\text{(\ref{eq:12834710239847})}}{<}
	L (\beta_0 \cdot n_0) \leq O(\log\log n) \cdot (T/n)^{1/2} \cdot n,
$$
where the last inequality comes from the fact that $\beta_0 = (T/n)^{1/2}$ due to Claim~\ref{cl:brecurse} and $n_0 = n$ by definition. Moreover, recall that $T \geq m+n\log^\alpha n$ for some large enough constant $\alpha$, therefore $(T/n)^{1/2} \geq \log^{\alpha/2} n \gg O(\log \log n)$. This means that
$$O(\log\log n) \cdot (T/n)^{1/2} \cdot n \ll (T/n)^{1/2} \cdot (T/n)^{1/2} \cdot n = T.$$
Therefore, the total increase over the potential space over the course of the algorithm is at most $T$, meaning that indeed for any $r$, $y_r = O(T)$ and thus in any iteration we have at most $O(T)$ edges. It is only left to prove Claim~\ref{cl:nrecurse}.

\begin{proof}[Proof of Claim~\ref{cl:nrecurse}]
	To prove the claim, we show that for every vertex of level $i-1$ that gets saturated and is marked as a leader, there are $(\beta_{i-1})^{0.5}$ other unique vertices of level $i-1$ that are not marked as a leader and are removed from the graph. This is clearly sufficient to show $n_i \leq n_{i-1} \cdot (\beta_{i-1})^{-0.5} \ll n_{i-1} \cdot (\beta_{i-1})^{-0.25}$.

	Consider some arbitrary iteration of the algorithm, and denote the set of saturated vertices and leaders with budget $\beta_{i-1}$ by $S$ and $L$ respectively. Since each saturated vertex of budget $\beta_{i-1}$ is chosen to be a leader independently with probability $\Theta(\frac{\log n}{\beta_{i-1}})$, we have $\E[|L|] = \Theta(\frac{\log n}{\beta_{i-1}} |S|)$. On the other hand, note that if $S \not = \emptyset$, we have $|S| \geq \beta_{i-1}$ since a vertex of budget $\beta_{i-1}$ is marked as saturated in Line~\ref{line:saturated1} of $\relabelintra$ if it has at least $\beta_{i-1}$ active neighbors with budget $\beta_{i-1}$, all of which will also get marked as saturated in Line~\ref{line:saturated2} and thus join $S$. Therefore, $|S| \geq \beta_{i-1}$, meaning that $
\E[|L|] = \Theta(\frac{\log n}{\beta_{i-1}}|S|) = \Omega(\log n).
$
Thus, by a standard Chernoff bound argument, we get $|L| = \Theta(\frac{\log n}{\beta_{i-1}}|S|)$ with high probability. On the other hand, recall that by Claim~\ref{cl:saturated}, every non-leader saturated vertex will be contracted to a leader in its 2-hop. That is, all vertices in $S \setminus L$ will be removed from the graph. This, averaged over the vertices in $L$, we get 
$$
\frac{|S \setminus L|}{|L|} \geq \frac{|S| - |L|}{|L|} \geq \frac{|S|}{|L|} - 1 \geq \frac{|S|}{\Theta(\frac{\log n}{\beta_{i-1}}|S|)} - 1 \geq \Theta\left(\frac{\beta_{i-1}}{\log n}\right)
$$
unique vertices that are removed from the graph per leader. Thus, it suffices to show that $\frac{\beta_{i-1}}{\log n} \gg (\beta_{i-1})^{0.5}$. For this, observe from Claim~\ref{cl:brecurse} and $T \geq m + n\log^\alpha n$ that $\beta_{i-1} \geq (T/n)^{1/2} \geq \log^{\alpha/2} n$ where $\alpha$ is some sufficiently large constant. It suffices to set $\alpha > 4$, say $\alpha = 5$, to get $\beta_{i-1} \gg \log^{2} n$ and thus $\log n \ll (\beta_{i-1})^{0.5}$. This indeed means
$
\Theta(\frac{\beta_{i-1}}{\log n}) \gg \frac{\beta_{i-1}}{(\beta_{i-1})^{0.5}} = (\beta_{i-1})^{0.5}
$
as desired.
\end{proof}

We already showed how proving Claim~\ref{cl:nrecurse} gives an upper bound of $O(T)$ on the potential space of all iterations, which by definition, is also an upper bound on the number of edges in the graph, concluding the proof of Lemma~\ref{lem:totalspace}.
\end{proof}

The next lemma is important for implementing the algorithm.

\begin{lemma}\label{lem:sumofbudgesquare}
	For any $r$, we have $\sum_{v \in V_r} (b_r(v))^2 \leq T$.
\end{lemma}
\begin{proof}
	We use induction on $r$. For the base case with $r=0$, we have $$\sum_{v \in V_0}(b_0(v))^2 = \sum_{v \in V_0}((T/n)^{1/2})^2 = T.$$ 
	Suppose by the induction hypothesis that $\sum_{v \in V_{r-1}}(b_{r-1}(v))^2 \leq T$, we prove that $\sum_{v \in V_r}(b_{r}(v))^2 \leq T.$ For this, it suffices to show that $\sum_{v \in V_r}(b_{r}(v))^2 \leq \sum_{v \in V_{r-1}}(b_{r-1}(v))^2$. Recall that we only increase the budgets in the $\relabelintra$ procedure, thus we only have to consider the effect of this procedure. Take a vertex $v \in V_r$ and with $b_{r}(v) > b_{r-1}(v)$ (otherwise the sum remains unchanged clearly). Note that $v$ must have been marked as a leader in iteration $r$ and thus $b_{r}(v) = b_{r-1}(v)^{\budgetconstant}$. Recall from the proof of Claim~\ref{cl:nrecurse} above that there are at least $b_{r-1}(v)^{0.5}$ unique vertices for $v$ with budget $b_{r-1}(v)$ that get removed from the graph in iteration $r$. Denote the set of these vertices by $U$. Removing these vertices decreases sum of budgets' square by 
	\begin{equation}\label{eq:8234234}
	\sum_{u \in U} b_{r-1}(u)^2 = |U| (b_{r-1}(v))^2 \geq b_{r-1}(v)^{2.5}.
	\end{equation}
	On the other hand, increasing the budget of $v$ from $b_{r-1}(v)$ to $b_{r-1}(v)^{1.25}$ increases the sum of budgets' square by $(b_{r-1}(v)^{1.25})^2 = b_{r-1}(v)^{2.5}$ which is not more than the decrease due to (\ref{eq:8234234}). Thus, we have $\sum_{v \in V_r}(b_{r}(v))^2 \leq \sum_{v \in V_{r-1}}(b_{r-1}(v))^2$ as desired.
\end{proof}

Finally, we have to argue that each iteration of Algorithm~\ref{alg:connectivity} can be implemented in $O(1)$ rounds of \MPC{} using $O(n^\delta)$ space per machine and with $O(T)$ total space. Since the proof is straightforward by known primitives, we defer it to Appendix~\ref{apx:implementation}.

\section{Improving Total Space to $O(m)$}\label{sec:shrinkvertices}

In the previous section, we showed how it is possible to find connected components of an input graph in $O(\log D + \log\log_{T/n} n)$ rounds so long as $T \geq m + n\log^\alpha n$. In this section, we improve the total space to $O(m)$. The key to the prove is an algorithm that shrinks the number of vertices by a constant factor with high probability. More formally:

\begin{lemma}\label{lem:shrink}
	There exists an \MPC{} algorithm using $O(n^\delta)$ space per machine and $O(m)$ total space that with high probability, converts any graph $G(V, E)$ with $n$ vertices and $m$ edges to a graph $G'(V', E')$ and outputs a function $f: V \to V'$ such that:
	\begin{enumerate}[ref={\arabic*},itemsep=-0.5ex]
		\item $|V'| \leq \gamma n$ for some absolute constant $\gamma < 1$.
		\item $|E'| \leq |E|$.
		\item For any two vertices $u$ and $v$ in $V$, vertices $f(u)$ and $f(v)$ in $V'$ are in the same component of $G'$ if and only if $u$ and $v$ are in the same component of $G$.
	\end{enumerate}
\end{lemma}

We emphasize that Lemma~\ref{lem:shrink} shrinks the number of vertices by a constant factor {\em with high probability}. This is crucial for our analysis. An algorithm that shrinks the number of vertices by a constant factor in expectation was already known  \cite{DBLP:conf/soda/KarloffSV10} but cannot be used for our purpose.

Let us first show how Lemma~\ref{lem:shrink} can be used to improve total space to $O(m)$ proving Theorem~\ref{thm:main}.

\begin{proof}[Proof of Theorem~\ref{thm:main}]
	First, observe that if $m \geq n \log^\alpha n$ or if $T \geq m + n\log^\alpha n$, then the algorithm of Section~\ref{sec:algorithm} already satisfiees the requirements of Theorem~\ref{thm:main}. Assuming that this is not the case, we first run the algorithm of Lemma~\ref{lem:shrink} for $(\alpha \log_{1/\gamma} \log n)$ iterations. Let $G'(V', E')$ be the final graph and $f$ be the function mapping the vertices of the original graph to those of $G'$. We have
	\begin{equation*}
	|V'| \leq n \cdot \gamma^{\alpha \log_{1/\gamma} \log n} = n \cdot \log^{-\alpha} n.
	\end{equation*}
	Now, we can run the algorithm of Section~\ref{sec:algorithm} on graph $G'$ to identify its connected components. The total space required for this is $$O(|E'| + |V'| \cdot \log^{\alpha} |V'|) = O\Big(m + (n \cdot \log^{-\alpha} n) \cdot \log^\alpha n\Big) = O(m+n) = O(m).$$
	We can then use function $f$ to identify connected components of the original graph in $O(1)$ rounds.
	
	Also, observe that the running time required is $O(\log \log n) + O(\log D + \log\log_{T/n} n)$. Given that $m \leq n\log^\alpha n$ and $T \leq m + n\log^\alpha n$ (as discussed above), we have $T/n = O(\poly\log n)$, thus $\log \log_{T/n} n = \Omega(\log \frac{\log n}{\log \log n}) = \Omega(\log \log n)$; meaning that $O(\log \log n) + O(\log D + \log\log_{T/n} n) = O(\log D + \log \log_{T/n} n)$  and thus the running time also remains asymptotically unchanged.
\end{proof}

We now turn to prove Lemma~\ref{lem:shrink}.

\begin{proof}[Proof of Lemma~\ref{lem:shrink}]
	In order to prove this lemma, we show that the following procedure reduces the number of vertices of the graph by a constant factor, with high probability. This procedure only merges some neighboring vertices and hence maintains the connected components. In this procedure, without loss of generality, we assume that there is no isolated vertices. One can simply label and remove all isolated vertices at the beginning.
	It is easy to implement this procedure in $\frac 1 {\delta}$ rounds using $O(n^{\delta})$ space per machine and a total space of $O(m)$ (see Appendix~\ref{apx:implementation} for implementation details).

\begin{tboxalg}{Shrinks the number of vertices by a constant factor in $O(1)$ rounds w.h.p.}\label{alg:shrinkvertices}
	\begin{enumerate}[ref={\arabic*}, topsep=0pt,itemsep=-0.5ex,partopsep=0ex,parsep=1ex, leftmargin=*]
		\item For each vertex $v$, draw a directed edge from $v$ to its neighbor with the minimum id. \label{line1}
		\item If for two vertices $u$ and $v$, we drew two directed edges $(u,v)$ and $(v,u)$, we remove one arbitrary.\label{line2}
		\item If a vertex has more than one incoming edge, we remove its outgoing edge.\label{line3}
		\item If a vertex $v$ has more than one incoming edge, we merge it with all its neighbors pointing to $v$ and remove the incoming directed edges of the neighbors of $v$.\label{line4}
		\item We remove each edge with probability $2/3$.\label{line5}
		\item We merge each directed isolated edge.\label{line6}
	\end{enumerate}
\end{tboxalg}

Next we show that this procedure reduces the number of vertices by a constant factor. For readability, we do not optimize this constant. 
Note that in Line \ref{line1} we are adding $n$ edges. It is easy to see that there is no cycle of length larger than $2$ in the directed graph constructed in Line \ref{line1}. Line \ref{line2} removes at most half of the edges. Moreover, it removes all cycles of length $2$. Thus by the end of Line \ref{line2} we have a rooted forest with at least $n/2$ edges. 

After Line \ref{line3} every vertex with indegree more than $1$ has no outgoing edges. Recall that each vertex has at most one outgoing edge. Thus, after Line \ref{line3} we have a collection of rooted trees where only the root may have degree more than $2$. We call such trees \emph{long tail stars}. Note that if we remove the outgoing edge of a vertex $v$ there are two incoming edges pointing to $v$ (which uniquely correspond to $v$). Although the process of Line \ref{line3} may cascade and remove the incoming edges of $v$, the following simple double counting argument bounds the number of removed edges. Note that this argument is just to bound the number of the edges and we do not require to run it, in order to execute our algorithm.

We put a token on each directed edge of the forest (before running Line \ref{line3}). Next we are going to move the tokens such that (a) we never have more than two tokens on each edge, and (b) at the end we move all tokens to the edges that survive after Line \ref{line3}. This says that at least half of the edges (i.e., at least $n/4$ edges) survive Line \ref{line3}.

We traverse over each rooted tree from the root to the leaves. At each step, if the outgoing edge of a vertex $v$ is removed, by induction hypothesis there are at most two tokens on the edge. Also, $v$ has at least two incoming edges. We move each of the tokens on the outgoing edge of $v$ to one of its incoming edges. Note that this is the only time we move a token to the incoming edges of $v$ and hence we do not have more than two tokens on each edge as desired.

If we merge a vertex $v$ with $r$ incoming edges in Line \ref{line4}, we remove at most $2r$ directed edges ($r$ incoming edges of $v$ and at most one incoming edge per each neighbor of $v$. On the other hand, we decrease the number of vertices by $r$. Thus, if this stage removes more than $n/8$ edges the number of vertices drops to at most $n-\frac n {16} = \frac{15}{16}n$, as desired. To complete the proof, we assume that at most $n/8$ edges are removed in Line \ref{line4} and show that in this case Lines \ref{line5} and \ref{line6} decrease the number of vertices by a constant factor.

Note that Line \ref{line4} removes the root of all long tailed stars. Thus after Line \ref{line4} we have a collection of directed edges. The probability that an edge passed to Line \ref{line5} becomes an isolated edge after sampling is at least $\frac 2 3 \cdot \frac 1 3 \cdot \frac 2 3 = \frac 4 {27}$. If we mark every third edge (starting from an end of each path), the chance that each marked edge becomes an isolated edge after sampling is independent of other marked edges. There are $\frac 1 3 \cdot \frac n 8 = \frac n {24}$ marked edges. Let $X$ be a random variable that indicates the number of marked edges that are isolated after sampling. Note that $\E[X]\geq \frac{4}{27}\frac{n}{24}=\frac n {162}$. By applying a simple Chernoff bound we have
\begin{align*}
\Pr\Big[X \leq  0.5 \frac n {162}\Big] \leq \exp \Big(-\frac{0.5^2 \frac n {162}}{2}\Big) 
=  \exp \Big(-\frac{n}{1296}\Big).
\end{align*}
Therefore, with high probability we merge at least $\frac{n}{324}$ edges in Line \ref{line6} as desired. 
\end{proof}

\section{Lower Bound}\label{sec:hardness}
In this section we show a conditional lower bound on the round complexity of finding connected components in the \MPC{} model.
We use the following conjecture to show our hardness result.
\begin{conjecture}[\twocycle{} conjecture~\cite{DBLP:conf/icml/YaroslavtsevV18, DBLP:conf/spaa/RoughgardenVW16, cc-contractions, DBLP:journals/corr/abs-1805-02974}]\label{conj:twocycle}
	Any \MPC{} algorithm that uses $n^{1-\Omega(1)}$ space per machine requires $\Omega(\log n)$ rounds to distinguish one cycle of size $n$ from two cycles of size $n/2$ with high probability.
\end{conjecture}

\newcommand{\alg}[0]{\ensuremath{\textsc{alg}}}

The conjecture above implies that the round complexity of our algorithm is tight for graphs with diameter $\Omega(n)$. However, it leaves the possibility of having faster algorithms for graphs with smaller diameter. For instance, one may still wonder whether for the case of graphs with $D = n^{1-\Omega(1)}$, an $O(1)$ connectivity algorithm exists or not. In what follows, we refute this possibility and show that the round complexity of our algorithm is indeed conditionally tight as long as $D = \log^{1+\Omega(1)} n$.
\hard*
\begin{proof}
	We prove this theorem by contradiction. Assume that there exists an algorithm $\alg$ in the \MPC~model with $n^{1-\Omega(1)}$ space per machine that finds all connected components of any given graph with diameter $D'$ w.h.p. in $o(\log D')$ rounds. Using this assumption we show that the following procedure applied to a graph consisting of sufficiently long disjoint cycles, shrinks the length of each cycle by a factor of $\frac {4\log n}{D'}$ w.h.p.
	
	\begin{tboxalg}{}
		\textbf{Input}: a graph $G$ consisting of disjoint cycles.
		\smallskip

	\begin{enumerate}[ref={\arabic*}, topsep=0pt,itemsep=-0.5ex,partopsep=0ex,parsep=1ex, leftmargin=*]
		\item Remove each edge of $G$ with probability $\frac {2\log n}{D'}$, thus obtaining $G'$.
		\item Find the connected components of $G'$ using $\alg$.
		\item Contract each connected component of $G'$ to a single vertex and return the obtained graph.
	\end{enumerate}
\end{tboxalg}

We first prove two properties of Algorithm~\ref{alg:shrinkvertices}, then show how it helps in obtaining a contradiction.

\begin{claim}
	Algorithm~\ref{alg:shrinkvertices} takes $o(\log D')$ rounds of \MPC{} with $n^{1-\Omega(1)}$ space per machine w.h.p.
\end{claim}
\begin{proof}
	Lines 1 and 3 of Algorithm~\ref{alg:shrinkvertices} can be trivially implemented in $O(1)$ rounds of \MPC{}. It suffices to show that the diameter of graph $G'$ is at most $D'$ so that running \alg{} takes $o(\log D')$ rounds.
	
	Fix a path of length $D'+1$ in $G$. The probability that all edges of this path survive Line 1 of Algorithm~\ref{alg:shrinkvertices} is $(1-\frac {2\log n}{D'})^{D'+1} \leq e^{-2\log n} = n^{-2}$. There are only $n$ such paths in $G$, thus by a simple union bound, w.h.p., none of them survives to $G'$; meaning that diameter of $G'$ is $\leq D'$.
\end{proof}

\begin{claim}\label{cl:removededges}
If $G$ has $m$ edges and $D' \leq m$, Algorithm~\ref{alg:shrinkvertices} removes at most $\frac{4 m \log n}{D'}$ edges w.h.p.
\end{claim}
\begin{proof}
Let $Z$ be the number of removed edges. We have $\E[Z]= \frac {2\log n }{D'} \cdot m$. Observe that $Z$ is sum of independent Bernoulli random variables and hence by Chernoff bound we have
	$$
		\Pr\Big[Z \geq 2\cdot\frac {2m\log n }{D'} \Big]
		\leq \exp \Big( -\frac{2m\log n }{3D'} \Big) 
		\leq \exp \Big( -\frac{2\log n }{3} \Big) 
		=n^{-2/3},
	$$
	as desired.
\end{proof}
	
	We iteratively run Algorithm~\ref{alg:shrinkvertices} and shrink the graph until it fits the memory of a single machine. Observe that after each application of Algorithm~\ref{alg:shrinkvertices}, only those edges that were removed from the graph will remain as the rest of the edges are contracted to single vertices. This means by Claim~\ref{cl:removededges} that if the current graph has $m$ edges, after one application of Algorithm~\ref{alg:shrinkvertices}, the resulting graph will have at most $\frac{4m\log n}{D'}$ edges. We repeat Algorithm~\ref{alg:shrinkvertices} for at most 
	\begin{align*}
\log_{\frac {D'} {4 \log n} }n = \frac{\log n}{\log(\frac {D'}{4\log n})} \leq \frac{\log n}{\log(\frac {D'}{\log n}) - \log 4} \leq \frac{\log n}{\log (D'^{\frac{\rho}{1+\rho}})-\log 4} &\leq \frac{\log n}{(\frac{\rho}{1+\rho})\log D' - \log 4} = O\bigg(\frac{\log n}{\log D'}\bigg)	
	\end{align*}
	times until the number of edges in the graph drops to $n^{o(1)}$ where we can store the entire graph on a single machine and solve the problem. The overall round complexity would be
	$
	O\big(\frac{\log n}{\log D'}\big) \cdot o(\log D) = o(\log n)
	$
	which is a contradiction.
\end{proof}

\section*{Acknowledgements}

We thank the anonymous FOCS reviewers for their comments. In addition, we thank S. Cliff Liu, Robert E. Tarjan, and Peilin Zhong \cite{DBLP:journals/corr/abs-2003-00614} for demonstrating an edge-case that we previously overlooked in our analysis; see Footnote~\ref{fnlabel} for a simple fix.


\appendix
\section{$\MPC{}$ Implementation}\label{apx:implementation}

In this section we provide some details on the implementation of
our algorithm in $\MPC$. We start by reviewing some known computational primitives in the $\MPC{}$ model with strictly sublinear space per machine.

\subsection{Primitives}
The following primitives can be implemented in the $\MPC{}$ model using $O(n^{\delta})$ space per machine.  All of the algorithms here
use space proportional to the input to the primitive (denoted by $N$).
\begin{itemize}
\item {\bf Sorting.} Sorting $N$ tuples be solved in
$O(1/\delta)$ rounds~\cite{DBLP:conf/isaac/GoodrichSZ11}. The input
is a sequence of $N$ tuples, and a comparison function $f$. The output
is a sequence of $N$ tuples that are in sorted order with respect to
$f$.

\item {\bf Filtering.} Filtering $N$ tuples can be solved in
$O(1/\delta)$ rounds. The input is a sequence of $N$ tuples, and a
predicate function $f$. The output is a sequence of $k$ tuples such
that a tuple $x$ is in the output if and only if $f(x) = \texttt{true}$.

\item {\bf Prefix Sums.} Computing the prefix sum of a sequence of $N$
tuples can be solved in $O(1/\delta)$
rounds~\cite{DBLP:conf/isaac/GoodrichSZ11}.
The input is a sequence of tuples $\{t_1, \ldots, t_N\}$, an
associative binary operator $\oplus$, and an identity element, $\bot$.
The output is a sequence of tuples $S$, s.t. $S_j$ is equal to $t_j$,
with an extra entry containing $\bot \oplus_{i=1}^{j-1} t_j$. Note that
reductions are a special case of prefix sums.

\item {\bf Predecessor.} The predecessor problem on a sequence of $N$
tuples can be solved in $O(1/\delta)$
rounds~\cite{DBLP:conf/isaac/GoodrichSZ11,andoniparallel}.
The input is a sequence of $N$ tuples, where each tuple has an
associated value in $\{0, 1\}$. The problem is to compute for each
tuple $x$, the first tuple $x'$ that appears before $x$ in the
sequence s.t. $x'$ has an associated value of $1$.

\item {\bf Duplicate Removal.} Given a sequence of $N$ elements, we
can remove duplicates in the sequence in $O(1/\delta)$ rounds by
simply sorting the elements and removing any tuple that is identical
to the one before it, using predecessor.
\end{itemize}

\subsection{Algorithm Implementation}
Here we show that each subroutine used in the algorithm can be
implemented in $O(1)$ rounds of $\MPC{}$. We start by describing the
representation of the data structures maintained by the algorithm.

\newpara{Data Representation}
Here we specify the representation of several data structures that we maintain over the course of the algorithm. 
All of the data structures are collections of tuples, which have a natural distributed
representation and can be stored on any machine. In addition to each tuple we store the
round that the tuple was written. This simplifies the process of applying updates to the sets.
\begin{itemize}
  \item $G=(V, E)$: The graph is represented as a set of vertex
  neighborhoods. Each vertex $u \in V$ stores its neighborhood, $N(u)$
  as a set of tuples $(u, v)$, which can be located on any machine.

  \item $b$: budgets are represented as a collection of tuples $(u,
  b(u))$. The levels, and active flags for each vertex are stored
  similarly.

  \item $C$: The set of vertices that have been merged to some vertex,
  $u$, are stored as a collection of $(u, x)$ tuples.

  \item We represent degrees implicitly by storing $O(\log \log_{T/n}
    n)$ entries for each vertex, $v$. The $i$'th entry indicates the
    number of level $i$ neighbors that $v$ has. We refer to this
    per-vertex structure as its degree array.
\end{itemize}

\newpara{Updating Budgets and Levels}
We update the budgets and levels as follows. We emit a
tuple $(v, b(v), r)$ where $r$ is the current iteration of the algorithm.
Updates can be processed by first sorting by decreasing lexicographic
order. Next, we can use predecessor and a filter to eliminate any
tuple $(v, b(v)', r')$ that is overwritten by a tuple $(v, b(v), r)$
where $r > r'$. This can be done in $O(1)$ rounds.

\newpara{Merging and Updating Neighbor Sets and Degrees}
As the algorithm proceeds we merge active vertices to other vertices
(for example in $\relabelinter{}$, when merging to our highest level
neighbor, and in $\relabelintra{}$, when merging to a leader in our
2-hop). We assume that the output of each merge operation is a tuple
$(x, m)$ indicating that a currently active vertex $x$ is merged to
$m$.

Let $m(x)$ be the id that $x$ is merging to. To merge vertices, we
map over all tuples $(u, v)$ representing the graph, to a set of
tuples $(u, v, 0)$. We also add the tuples $(x, m, 1)$ for each merge
input. Then we sort by the first entry in the tuple, and run
predecessor, which associates each $(u, v, 0)$ tuple with $m(x)$, and
lets us emit a collection of $(m(x), v, 0)$ tuples. We apply the same
idea again on the output of the previous step, with the first and
second entries swapped, which produces a collection of $(m(v), m(x),
0)$ tuples (the tuples with 1 can be filtered out). The graph on the
merged vertices is produced by swapping the components, and removing
any duplicate tuples that result from merges. We also remove
self-loops by filtering all $(u, v)$ tuples in $G$ where $u = v$.

Lastly, we can update the degree arrays by recomputing them after each
update and merge step. Updating is easily done by sorting, and
applying a prefix sum to produce the number of level $i$ neighbors
incident to each vertex.

Since we apply a constant number of $O(1)$ round algorithms, the
procedure to update all vertex sets and degrees in the graph takes
$O(1)$ rounds in total.

\newpara{Computing Degrees}
Note that the degree of all vertices can be computed by sorting all
$(x, y)$ tuples in the graph and applying a prefix sum. Computing the
degree of each vertex can therefore be done in $O(1)$ rounds.
Alternately, we can prefix sum the degree array for the vertex.

Since we maintain the level $i$ degree of each vertex explicitly in
the degree arrays, we can easily compute the induced degree of a
neighboring vertex when restricted to vertices with level $\geq i$.
This is done by applying a prefix sum over the tuples with degree
at least $i$.

\subsection{Implementing Algorithm~\ref{alg:connectivity}}\label{apx:connectivityimpl}

\newpara{Implementing $\addtwohop{}$}
In $\addtwohop{}$, each \emph{active} vertex, $v$, either fully
connects itself to its 2-hop if the size of its 2-hop has size at most
$b(v)$, or connects itself to $b(v) - d(v)$ neighbors arbitrarily.
First for each vertex $u \in N(v)$ we compute its degree when restricted to
vertices with level at least $\ell(v)$. This can be done in $O(1)$
rounds as described previously. Note that if the restricted degree of $v$ is more than $b(v)$ we do not need to add any edges since $b(v) - d(v) \leq 0$. Thus, we assume $d(v) < b(v)$.

{\bf Case 1.} If any of $v$'s neighbors has a restricted degree $\geq b(v)$, we
take the first $b(v)$ vertices from this neighbors degree and union
them with the vertices currently in $N(v)$. We mark each of our
current neighbors with $0$ and mark each new (incoming vertex) with
$1$ and remove duplicates (ignoring the $\{0, 1\}$ tag on the
tuples). Next, we sort lexicographically. If this set contains more
than $b(v)$ vertices, we pick the first $b(v)$ to include in $N(v)$
and drop the remaining tuples. Note that any neighbors of $v$ are
guaranteed to remain in this set, since $d(v)$ was initially less than
$b(v)$, and the lexicographic sort will order our existing neighbors
before any new neighbors. The total space used in this step is
$O(b(v))$ per vertex, and the output is exactly $b(v)$ neighbors with
level at least $l(v)$.

{\bf Case 2.} If each of $v$'s neighbors has restricted degree $< b(v)$, since
$v$ has less than $b(v)$ neighbors of level at least $\ell(v)$, and each has restricted degree
smaller than $b(v)$, we can copy each of these neighbor lists into the
space available for $v$, which is at least $\Theta(b(v)^{2})$.
Finally, we can remove duplicates for these neighbors. By using a
similar tagging idea as in the previous step, we tag each neighbor
based on whether it was present in $N(v)$ originally. If more than
$b(v)$ vertices are produced in this step, we pick the first $b(v)$ of
them in the lexicographically ordered sequence.

As we use a constant number of primitive calls, each of which
require $O(1)$ rounds, the overall round-complexity is $O(1)$.
Furthermore, the maximum space a vertex uses is at most $O(b(v)^{2})$, which by Lemma~\ref{lem:sumofbudgesquare}, is precisely the space we can use for each vertex while using only $O(T)$ total space. The output either
fully connects $v$'s 2-hop, or updates $N(v)$ to have size at most
$b(v)$, adding vertices chosen arbitrarily from $v$'s
degree-restricted 2-hop.

\newpara{Implementing $\relabelinter{}$}
In $\relabelinter{}$, each \emph{active} vertex, $v$, chooses
the highest level vertex $h(v)$ in its direct neighborhood and merges
itself to it.

We can implement this procedure by having each active vertex sort its
direct neighbors by their level in descending order. If the first
neighbor in this sequence has the same level as $v$, we do nothing,
otherwise we have found $h(v)$ and merge $v$ to $h(v)$ by emitting a
tuple $(v, h(v))$. We then merge and update the neighbor sets of all
vertices in $G$. Overall, the algorithm runs in $O(1)$ rounds.

\newpara{Implementing $\relabelintra{}$}
In $\relabelintra{}$, each active saturated vertex, $v$, first samples
itself to be a leader. If it is chosen as a leader, it does nothing.
Otherwise, each non-leader saturated vertex selects a leader in its
2-hop and merges with it. The sampling can easily be done in $1$ round
of computation, assuming that each machine has a source of randomness.
For each sampled leader, we update the level and budget of the vertex
by writing the tuple $(v, l(v) + 1)$ to the levels and $(v,
b(v)^{\budgetconstant})$ to the budgets. These are updated as
described earlier.

Next, we must check whether a vertex has a leader in its 2-hop, which
can be implemented as follows. First, each active saturated vertex
chooses a candidate leader in its neighborhood, breaking ties
arbitrarily if multiple leaders are present. If no candidate exists,
we mark this fact with a null value. This can be done by using mapping
the graph to tuples $(u, v, l)$ where $l \in \{0, 1\}$ indicates
whether $v$ is marked as a leader. Then, for each $u$ we inject a
tuple $(u, \infty, 0)$, perform a lexicographic sort, and compute
predecessor. Each $(u, \infty, 0)$ tuple finds the first tuple before
it that contains a $1$---if this tuple's first entry starts with $u$,
we use the second entry as the chosen candidate for $u$. Otherwise,
the candidate is set to null. The candidates are a collection of $(u,
c)$ pairs, where $c$ is the candidate for vertex $u$ and null
otherwise. Note that the algorithm just described computes a function
$f$ (in this case projecting a leader) which is aggregated over the
neighbors of a vertex in $O(1)$ rounds.

Lastly, each non-leader saturated vertex performs another aggregation,
identical to the one described in the previous step, which gives each
active saturated vertex $v$ a leader, $l(v)$, in its 2-hop w.h.p.
We emit tuples $(v, l(v))$ indicating that $v$ is merged with $l(v)$.
Finally we merge and update the neighbor sets of all vertices in $G$.
The algorithm runs in $O(1)$ rounds as it performs a constant number
of steps, each of which take $O(1)$ rounds.

\subsection{Implementing Algorithm~\ref{alg:shrinkvertices}}
We now discuss how each subroutine used in
Algorithm~\ref{alg:shrinkvertices} can be implemented. Recall that
this algorithm eliminates a constant factor of the vertices in the
graph w.h.p. in $O(1)$ rounds of $\MPC{}$.

Line~\ref{line1} can be implemented by using a reduction (prefix sum
with $\min$) over the neighbors of each vertex. Line~\ref{line2} can
be implemented by sorting the chosen edges, and removing duplicates.
Note that if both $(u, v)$ and $(v, u)$ are chosen (say $u < v$), only
the $(u, v)$ edge remains. To implement Line~\ref{line3}, we first
compute the in-degree of each vertex, which can be done by sorting.
Next, we send the in-degrees of each vertex to its outgoing edge,
which can be done via sorting and predecessor, and drop the outgoing
edge if its in-degree is greater than 1. To implement
Line~\ref{line4}, each $v$ with incoming edge set $\{(v, u_1), \ldots,
(v, u_k)\}$, we generate the tuples $\{(u_1, v), \ldots, (u_k, v)\}$
which indicates that $u_1, \ldots, u_k$ should be merged to $v$.
Recall that we can use a previously described merging algorithm to
merge these vertices to $v$ in $O(1)$ rounds. Line~\ref{line5} simply
drops each tuple for a remaining edge with probability $2/3$. Lastly,
in Line~\ref{line6} we can detect isolated edges by computing the
in-degree and out-degree of the vertices as previously described,
summing them together, and choosing out edges of vertices whose
in-degree and out-degree sum to $1$, which can be done in $O(1)$
rounds in total. As each step takes $O(1)$ rounds of $\MPC{}$, each
call to Algorithm~\ref{alg:shrinkvertices} takes $O(1)$ rounds in
total.

\section{PRAM implementation}

In this section we show that our connectivity algorithm can be simulated in
$O(\log D + \log\log_{m/n} n)$ depth on the \multipram{}, a strong model of parallel
computation permitting concurrent reads and concurrent writes. The parallel
algorithm we derive performs $O((m+n)(\log D + \log\log_{m/n} n))$ work and is
therefore nearly work-efficient. We start by describing existing \PRAM{} models,
how the \multipram{} compares to these models, and reviewing existing results on
parallel graph connectivity algorithms from the literature.

\subsection{Model}
We state results in this section in the \emph{work-depth model} where the \emph{ work} is equal to the number of operations required (equivalent to the
processor-time product) and the \emph{depth} is equal to the number of time steps
taken by the computation. The related machine model used by our algorithms is
the parallel random access machine (\PRAM{}). Note that in work-depth models, we
do not concern ourselves with how processors are mapped to tasks (see for
example, Jaja~\cite{DBLP:books/aw/JaJa92} or Blelloch et al.~\cite{Blelloch96,
  Blelloch:2010:PA:1882723.1882748}). We now place our machine model, the
\multipram{}, in context by reviewing related PRAM models.

The \arbpram{} handles concurrent writes to the same memory cell by selecting an
arbitrary write to the cell to succeed. The \scanpram{} extends the \arbpram{}
with a unit-depth scan (prefix-sum) operation~\cite{DBLP:conf/icpp/Blelloch87}
(note that in the original paper the extended model was the EREW \PRAM{}). The
inclusion of this primitive is justified based on the observation that a
prefix-sum can be efficiently implemented in hardware as quickly as retrieving a
reference to shared-memory. The \combpram{} combines concurrent writes to the
same memory location based on an associative and commutative combining operator
(e.g., sum, max). The \multipram{} extends the \arbpram{} with a unit-depth
\emph{multiprefix operation} which is a generalization of the scan operation
that performs multiple independent scans. The input to the multiprefix operation
is a sequence of key-value pairs. The multiprefix performs an independent scan
for each key, and outputs a sequence containing the result of each
scan. The \multipram{} was proposed by Ranade, who gave a routing
algorithm for butterfly networks in which a multiprefix could be
implemented as quickly as fetching a memory
reference~\cite{Ranade:1989:FPC:916125}. In all \PRAM{} models
considered in this paper we assume that each processor has its own
random source. We refer the interested reader to the Karp and
Ramachandaran chapter on parallel algorithms for more details on
\PRAM{}s~\cite{KarpR90}.

We observe that all of aforementioned \PRAM{} models can be
work-efficiently
simulated in the \MPC{} model with strictly sublinear space per machine, such that the number of rounds of the resulting \MPC{} computation is asymptotically equal to the depth. To see
this, note that the \multipram{} can work-efficiently simulate all of the other
\PRAM{} variants, without an increase in depth. Furthermore, a multiprefix
operation on $n$ key-value pairs can be implemented in the \MPC{} with strictly
sublinear space per machine in $O(n)$ space and $O(1)$ rounds by performing
independent scan operations for each key independently in parallel (see Section
E.6 (Multiple Tasks) in~\cite{andoniparallel} for implementation details).
Therefore, the \MPC{} model with strictly sublinear space per machine is more
powerful than the \PRAM{} variants described above.

\subsection{Parallel Connectivity Algorithms}
Connectivity algorithms on the \PRAM{} have a long history, and many
algorithms have been developed over the past few
decades~\cite{shiloach1982connectivity, awerbuch1983connectivity,
  reif1985cc, phillips1989contraction, gazit1991connectivity,
  DBLP:conf/spaa/KargerNP92, cole96minimum, halperin1994cc,
  halperin2000cc, poon1997msf, pettie2000mst, shun2014practical}.
Classic parallel connectivity algorithms include the hook-and-contract
algorithms of Shiloach and Vishkin~\cite{shiloach1982connectivity} and
Awerbuch and Shiloach~\cite{awerbuch1983connectivity}, and the
random-mate algorithms of Reif~\cite{reif1985cc} and
Phillips~\cite{phillips1989contraction}. All of these algorithms
reduce the number of vertices in each round by a constant fraction,
but do not guarantee that the number of edges reduces by a constant
fraction, and therefore perform $O(m\log n)$ work and run in $O(\log
n)$ depth on the \multipram{}. (The algorithms of Reif and Phillips
are randomized, so the bounds hold w.h.p.)

Historically, obtaining a work-efficient parallel connectivity algorithm
(an algorithm which performs asymptotically the same work as the most
efficient sequential algorithm) was difficult, and progress was not
made until the early 90s~\cite{gazit1991connectivity}. A number of
work-efficient algorithms were subsequently
discovered~\cite{cole96minimum, halperin1994cc, halperin2000cc,
poon1997msf, pettie2000mst, shun2014practical}. Many of these
work-efficient algorithms also achieve $O(\log n)$ depth, with the
algorithm of Halperin and Zwick achieving this bound on the EREW
\PRAM{}~\cite{halperin1994cc}.

Since a work bound of $O(m+n)$ is optimal, the remaining question is
whether the depth can be improved.  In terms of lower-bounds, the
results of Cook \etal{}~\cite{Cook:1986:ULT:13535.13541} and
Dietzfelbinger \etal{}~\cite{Dietzfelbinger:1994:ELT:185297.185303}
imply a lower bound of $\Omega(\log n)$ depth for connectivity on
randomized EREW \PRAM{}s. Despite the lower-bound only holding for
the EREW setting, to the best of our knowledge all existing
connectivity algorithms in the \PRAM{} literature run in least
$\Omega(\log n)$ depth, even in models permitting concurrent reads and
writes. A natural question therefore is whether we can solve
connectivity on a stronger \PRAM{} model in $o(\log n)$ depth.

If work-inefficiency is permitted, the answer is
certainly yes. For example, on the CRCW \PRAM{}, which permits
unbounded fan-in writes, a folklore result for connectivity is to
perform matrix squaring, stopping once each connected component
becomes a clique. This algorithm has $O(\log D)$ depth, but requires
$O(n^{3})$ work using a combinatorial matrix multiplication algorithm
and is therefore work-inefficient for sparse graphs. Note that
$O(\log D)$ depth is a natural goal for the \PRAM{}, and is in fact a
lower bound if the \twocycle{} conjecture is true, due to known
simulations of the \PRAM{} on
\MPC{}~\cite{DBLP:conf/soda/KarloffSV10,
DBLP:conf/isaac/GoodrichSZ11}. Is there a more work-efficient parallel
connectivity algorithm that runs in $O(\log D)$ depth? Our result
shows that a nearly work-efficient \PRAM{} algorithm with $O(\log D +
\log\log_{m/n} n)$ depth w.h.p.  exists in the \multipram{}, resolving
this question affirmatively for graphs with $D =
\Omega(\mathsf{polylog}(n))$.

\subsection{Multiprefix CRCW Implementation of Algorithm~\ref{alg:connectivity}}\label{apx:paralleconnimpl}

\newpara{Data Representation}
We represent \emph{active} and \emph{next} as dense arrays of length $n$.  We
represent $C$ using an array for each vertex. We store the graph in a sparse
format, storing each vertex's neighbors in an array.

\newcommand{\defn}[1]{\emph{\textbf{#1}}}
\newpara{Sequence Primitives}
\begin{itemize}
\item \defn{Map} takes as input an array $A$ and a function $f$ and
applies $f$ to each element of $A$. Map can be implemented in $O(n)$
work and $O(1)$ depth on all of the \PRAM{} models considered in this
paper.

\item
\defn{Scan} takes as input an array $A$ of length $n$, an associative binary
operator $\oplus$, and an identity element $\bot$ such that $\bot \oplus x = x$
for any $x$, and returns the array $(\bot, \bot \oplus A[0], \bot \oplus A[0]
\oplus A[1], \ldots, \bot \oplus_{i=0}^{n-2} A[i])$ as well as the overall sum,
$\bot \oplus_{i=0}^{n-1} A[i]$. We use \defn{plus-scan} to refer to the scan
operation with $\oplus = +$ and $\bot = 0$. Scan can be implemented by simply
dropping the keys (or making all keys equal) and applying a multiprefix.

\item
\defn{Multiprefix} takes as input an array $A = [(k_0, v_0), \ldots, (k_{n-1},
v_{n-1})]$ of length $n$, an associative binary operator $\oplus$ and an
identity element $\bot$ (similarly to scan), and returns an array $[(k_0, v'_0),
\ldots, (k_{n-1}, v'_{n-1})]$ where the output values associated with each key
are the result of applying an independent scan operation for the values with
each key. We use a \defn{plus-multiprefix} to refer to the multiprefix operation
with $\oplus = +$ and $\bot = 0$.

\item
\defn{Remove Duplicates} takes an array $A$ of elements and returns a new array
containing the distinct elements in $A$, in the same order as in $A$.
Removing duplicates can be implemented by using a plus-multiprefix operation
where the keys are elements, and the values are initially all $1$s. Since
the plus-multiprefix assigns the first instance of each key a $0$ value, the
keys corresponding to values greater than $0$ are filtered out, leaving only a
single copy of each distinct element in $A$.

\item
\defn{Filter} takes an array $A$ and a predicate $f$ and returns a new array
containing $a \in A$ for which $f(a)$ is true, in the same order as in $A$.
Filter can be implemented by first mapping the array in parallel with the
predicate $f$, and setting the key to $1$ if $f(e)$ is true, and $0$ otherwise.
A plus-multiprefix is then used to assign each element with a $1$ key contiguous
and distinct indices. Finally, the elements where $f(e)$ is true are copied to
the output array using the indices from the previous step.
\end{itemize}

The multiprefix operation on an array of length $n$ costs $O(n)$ work
and $O(1)$ depth on the \multipram{}. Therefore, scan, filter and
removing duplicates on arrays of length $n$ can also be implemented in
$O(n)$ work and $O(1)$ in this model.

We note that for convenience, the parallel algorithm we describe below
often runs \emph{multiple} multiprefix operations in parallel. These
operations can usually be run using a single multiprefix.  In
particular, parallel multiprefix operations can be simulated using a
single operation so long as each parallel operation is keyed by a
unique key. The idea is to prepend the unique key to the keys within
each multiprefix operation. In our implementation, this unique key is
usually the vertex id.

\newpara{Graph Primitives}
\defn{Symmetrize} takes as input a directed graph $G=(V, E)$ as a
collection of tuples and outputs an undirected graph $G_{S}$ in
adjacency array form. The algorithm first computes a plus-scan over
the array $S = [2d(i) | i \in [0, n)]$. Each vertex $v$ then copies
its incident out-edges into an array of size $2m$ at offset $S[v]$,
and copies the edges with their direction reversed at offset $S[v] +
d(v)$. Next, the algorithm removes duplicates from the array. Lastly,
the algorithm collects the edges incident to a vertex contiguously.
This is done by first running a plus-multiprefix, $M$, where the keys are
the first component of each tuple and the value is $1$. Next, the
algorithm computes $S$, a plus-scan over the distinct keys, where the
value is degree of the vertex, which is obtained from the result of
the multiprefix within each key. The algorithm finally allocates an
array proportional to the output of this scan, and copies edge $i$
into the location $S[fst(i)] + M[i]$ where $fst(i)$ is the first
component of edge $i$. The work of this operation is $O(|E| + |V|)$
and the depth is $O(1)$.

\defn{Contract} takes as input an undirected graph $G=(V, E)$ in
adjacency array form and a mapping $m : V \rightarrow V$ s.t.  either
$m(u) = u$ or $m(u) = v$ where $v \in N_{G}(u)$. The contraction algorithm
constructs the graph $G' = \{(m(u), m(v))\ |\ (u, v) \in
E(G)\}$ with duplicate edges and self-loops removed.  The contraction
algorithm is implemented as follows. The algorithm first uses a
plus-scan to count the number of remaining vertices, $n'$. Next, it
computes $d'(v)$, the degree of the $v$ in $G'$ using a
plus-multiprefix (or a combining write, which can be work-efficiently
simulated in $O(1)$ depth). It then maps each edge $(u, v) \in E(G)$
to $(m(u), m(v))$, and uses a multiprefix operation to remove
duplicates, which takes $O(m)$ work and $O(1)$ depth. After this step,
edges exist in both directions and all edges incident to a vertex
$m(u)$ are stored contiguously. The work of this operation is $O(|E| +
|V|)$ and the depth is $O(1)$.

\newpara{Implementing $\addtwohop{}$}
Recall that in this step every active vertex either fully connects to its 2-hop
if the size of its 2-hop is at most $b(v)$, or connects to $C = \min(b(v)-d(v),
|H(v)|)$ neighbors arbitrarily. For ease of discussion we assume that we must
connect to exactly $b(v)$ neighbors.

{{\bf Case 1:} $\forall u \in N(v), d(u) \leq b(v)$.} In this case, the
algorithm must set $v$'s neighbors to $N(v) \cup \cup_{u \in N(v)} N(u)$. One
simple idea (implementable on the \arbpram{}) is to initialize a parallel
hash-table and insert all neighbors into the table, which takes $O(b(v)^2)$ work
and $O(\log^{*} n)$ depth~\cite{DBLP:conf/focs/GilMV91}. After insertion, all
elements in the table will be distinct. However, using parallel hashing
increases the depth by a multiplicative factor of $O(\log^{*} n)$. Instead, our
algorithm uses the multiprefix operation to copy the neighbor's neighbors into
an array and remove duplicates from this array in $O(1)$ depth.

Concretely, the algorithm first writes the degree of each of its neighbors
(including itself) into an array $S$, and computes a plus-scan over this array.
Next, it allocates an array $E_v$ with size proportional to the result of the
scan, and copies the neighbors of the $i$th neighbor to the sub-array
$E_v[S[i]]$. Note that the array $E_v$ now contains $N(v) \cup \cup_{u \in N(v)}
N(u)$, possibly with duplicates. We produce the new neighbors of $v$ by removing
duplicates from $E_v$.

We now address how to add edges discovered by a vertex $u$ in the
\addtwohop{} procedure the endpoint, $v$. Note that $v$ may already
have degree at least $b(v)$, but still have edges added to it by
vertices that discover $v$ in this procedure. This operation is
implemented by simply symmetrizing the graph.

{{\bf Case 2:} $\exists u \in N(v)$ s.t. $\ell(u) > \ell(v)$.}
This case can be checked in constant depth using a concurrent write. Note that
if $v$ has any higher-level neighbors it will become inactive on this round in
\relabelinter{} so we can quit.

{{\bf Case 3:} $\exists u \in N(v)$ s.t. $d(u) > b(v)$.}
This case can also be checked using an arbitrary write. In this case,
the algorithm copies $b(v)$ neighbor ids from the neighbor $v$ with
degree $> b(u)$.  It then removes duplicates, and adds edges from the
chosen endpoints to itself by symmetrizing the graph, as before.

\newpara{Implementing $\relabelinter{}$}
Recall that in $\relabelinter{}$, each \emph{active} vertex, $v$, chooses the
highest level vertex $h(v)$ in its direct neighborhood and merges itself to it
by updating \emph{next}. The maximum value can be selected either using a scan
with the $\max$ operation. An arbitrary neighbor with the maximum level in $v$'s
direct neighborhood can then be selected using a concurrent write. We then
contract the graph using the contraction primitive where $m = \emph{next}$.
Finally, we update $C$. Since each vertex is uniquely stored in $C(v)$, we do
not need to remove duplicates in this step and simply flatten the sets $C(v)
\forall m(v) = u$ to be contiguous, which can be done using a plus-scan and a
parallel copy.

\newpara{Implementing $\relabelintra{}$}
Detecting whether an active vertex, $v$ is saturated is done by computing a
prefix sum over its neighbors, filtering out neighbors with degree less than
$b(v)$. In a second synchronous step, each vertex $v$ checks whether it has a
neighbor $u$ with $l(u) = l(v)$ that is marked as saturated, and if so marks
itself as saturated by performing an arbitrary write. We then use the
processor's internal randomness to sub-sample vertices as leaders. Vertices
which successfully become leaders do nothing. The non-leader saturated vertices
select a leader in their 2-hop as follows.

First, each vertex checks if it has a leader in its direct neighborhood, which
can be done by mapping over the neighbors and using a concurrent write. Vertices
that successfully find a leader in their neighborhood indicate mark themselves
with the selected neighbor. Each vertex that failed the previous step re-check
their direct neighborhood, and pick an arbitrary leader chosen by some marked
neighbor, again using a concurrent write (such a neighbor exists w.h.p.).
Finally, we contract the graph using the contraction primitive, and update $C$
using the same method as in \relabelinter{}.

\subsection{Parallel Implementation of Algorithm~\ref{alg:shrinkvertices}}
Recall that this algorithm eliminates a constant factor of the
vertices in the graph w.h.p. per round. We discuss how to implement
each step of the algorithm on the \arbpram{}. Step~\ref{line1} can be
implemented work-efficiently using the minimum algorithm in $O(d(v))$
work and $O(1)$ depth w.h.p. Step~\ref{line2} can be implemented in
$O(n + m)$ work and $O(1)$ depth by checking each neighbor. Similarly,
Step~\ref{line3} can be implemented in two PRAM rounds by first
arbitrarily writing any neighbor pointing to the vertex $u$, and in
the second round writing another edge if it differs from the first
one. The step also takes $O(m+n)$ work and $O(1)$ depth.
Step~\ref{line4} can be checked similarly. The merge is handled using
the contraction primitive given above which costs $O(m+n)$ work and
$O(1)$ depth. Step~\ref{line5} can be done in $O(m + n)$ work and
$O(1)$ depth using a random source within each processor. Finally, we
can detect isolated edges in Step~\ref{line6} similarly to
Step~\ref{line3} above, and merge these edges using graph contraction
in $O(m+n)$ work and $O(1)$ depth using the contraction algorithm
described above. In total, one round of
Algorithm~\ref{alg:shrinkvertices} costs $O(m+n)$ work and $O(1)$
depth.

\subsection{Cost in the \multipram{} Model}

\workdepthpram*
\begin{proof}
Observe that in our parallel implementation, in each iteration a vertex
never performs more than $O(b(v)^{2})$ work. By
Lemma~\ref{lem:sumofbudgesquare} the total work is therefore $O(T)$
per round. All other operations in an iteration such as contraction
and symmetrization cost $O(m + n)$ work. As the algorithm performs
	$O(\log D + \log\log_{m/n} n)$ rounds w.h.p., the overall work is $O(T(\log
	D + \log\log_{m/n} n)) = O((m+n)(\log D + \log\log_{m/n} n))$ w.h.p. for $T =
O(m+n)$.

	By Lemma~\ref{lem:rounds} the overall depth is $O(\log D + \log\log_{m/n} n)$ w.h.p. since each round
of the algorithm is implemented in $O(1)$ depth.
\end{proof}

\bibliographystyle{plain}
\bibliography{references,referencesMPC}
\end{document}